\documentclass[11pt,reqno]{amsart}
\usepackage[T1]{fontenc}
\usepackage{lmodern}
\usepackage[a4paper,margin=28mm,headheight=14pt]{geometry}
\usepackage{amsmath,amssymb,amsthm,mathtools}
\usepackage{microtype}
\usepackage{booktabs,longtable,array}
\usepackage{enumitem}
\usepackage[sort,nocompress]{cite}
\usepackage{xurl}
\usepackage[hidelinks]{hyperref}
\makeatletter
\renewcommand{\@setauthors}{%
  \begin{center}
    \footnotesize\scshape
    \vspace{1ex}
    \authors\par
    \vspace{1.2ex}
    \normalfont\footnotesize
    \def\address##1##2{##2\par}%
    \addresses
    \vspace{0.5ex}
    \PaperCorrespondence\par
  \end{center}%
  \global\let\addresses\@empty
}
\makeatother
\numberwithin{equation}{section}
\newtheorem{theorem}{Theorem}[section]
\newtheorem{proposition}[theorem]{Proposition}
\newtheorem{lemma}[theorem]{Lemma}
\newtheorem{corollary}[theorem]{Corollary}
\theoremstyle{definition}
\newtheorem{definition}[theorem]{Definition}
\theoremstyle{remark}

\newcommand{\F}{\mathbb F_2}
\newcommand{\Mat}{M_3(\F)}
\newcommand{\T}{T_{\langle3,3,3\rangle}}
\DeclareMathOperator{\rank}{rank}
\DeclareMathOperator{\tr}{tr}
\DeclareMathOperator{\GL}{GL}
\DeclareMathOperator{\im}{im}
\DeclareMathOperator{\Span}{span}
\newcommand{\id}{\mathrm{id}}
\newcommand{\code}[1]{\texttt{\small #1}}
\newcommand{\qbinom}[2]{\genfrac{[}{]}{0pt}{}{#1}{#2}_2}
\setlist[enumerate]{label=\textup{(\roman*)},leftmargin=*,itemsep=3pt}
\title[A structural lower bound for binary matrix multiplication]
{A structural proof of the lower bound 21\\
for $3\times3$ matrix multiplication over $\mathbb F_2$}
\author[Shuxing Yang et al.]{\mbox{Shuxing Yang}, \mbox{Rui Zhao},
\mbox{Junyao Wu}, \mbox{Yize Wang}, \mbox{Wenhao Li}, \mbox{Fujia Chen},
\mbox{Taowen Deng}\\
\mbox{Shenzhan Hong}, \mbox{Yaqi Li}, \mbox{Zichen Li},
\mbox{Jincheng Mi}, \mbox{Yuang Pan}, \mbox{Kaihao Zhu}, \mbox{Junjie Yang}\\
\mbox{Hongsheng Chen\textsuperscript{*}},
\mbox{Yihao Yang\textsuperscript{*}}}
\address{College of Information Science and Electronic Engineering,
Zhejiang University\\Qiushi Engine Team}
\newcommand{\PaperCorrespondence}{\textsuperscript{*}Corresponding authors:
\mbox{Hongsheng Chen (\href{mailto:hansomchen@zju.edu.cn}{hansomchen@zju.edu.cn})};
\mbox{Yihao Yang (\href{mailto:yangyihao@zju.edu.cn}{yangyihao@zju.edu.cn})}.}
\date{}
\subjclass[2020]{Primary 68Q17; Secondary 15A69, 05B25, 68V20}
\keywords{Matrix multiplication; algebraic complexity; tensor rank;
bilinear complexity; finite fields; computer-assisted proof; formal verification}
\hypersetup{
  pdftitle={A structural proof of the lower bound 21 for 3x3 matrix multiplication over F2},
  pdfauthor={Shuxing Yang, Rui Zhao, Junyao Wu, Yize Wang, Wenhao Li, Fujia Chen, Taowen Deng, Shenzhan Hong, Yaqi Li, Zichen Li, Jincheng Mi, Yuang Pan, Kaihao Zhu, Junjie Yang, Hongsheng Chen, Yihao Yang},
  pdfsubject={A structural proof of the binary matrix multiplication rank lower bound 21}}
\begin{document}
\begin{abstract}
We prove that the tensor rank of $3\times3$ matrix multiplication over
$\mathbb F_2$ is at least $21$. The structural proof, independently
developed by Qiushi Engine, converts occupation constraints on a single
tensor factor into algebraic relations coupling all three factors.
Certified quotient-rank bounds and finite geometry force any
hypothetical $20$-term decomposition to have first-factor matrix-rank
profile $(16,1,3)$. The ranks of the corresponding split-flattened
summands therefore sum to $27$, exactly the rank of the full split
flattening. Equality in rank subadditivity forces their images to form
a direct sum; normalization by the inverse flattening then makes the
summands pairwise annihilating idempotents. An explicit product identity
for matrix multiplication implies that at most one first factor can be
invertible, contradicting the three forced by the profile. The same
obstruction constrains $22$-term decompositions attaining the
split-rank bound. The complete proof, including the finite quotient
bounds, is formalized in Lean. The accompanying research trajectory
records Qiushi Engine's long-horizon autonomous research, from numerical
experiments and quotient constructions to the structural proof.
\end{abstract}
\maketitle
\raggedbottom
\section{Introduction}\label{sec:introduction}

Matrix multiplication is a basic example of a bilinear map whose
algebraic structure permits fewer products than its entrywise formula
suggests. Strassen's seven-product algorithm for $2\times2$ matrices
established this possibility~\cite{strassen}. For $3\times3$ matrices,
the usual formula uses $27$ products, whereas Laderman's algorithm uses
$23$~\cite{laderman}. The minimum remains unknown. Its determination
asks both for efficient constructions and for structural obstructions
that apply to every possible bilinear algorithm.

A construction specifies linear forms in both inputs and the
coefficients combining their products into the output. Quotient-rank
bounds constrain a smaller set of data: the positions and
multiplicities of the first-input forms. In this paper we combine
these constraints with a split flattening of the complete
multiplication tensor. At a hypothetical length of twenty, the two
sets of conditions force equality in matrix-rank subadditivity.
The equality case then yields relations among all three factor lists,
strong enough to exclude the decomposition.

Write $\T$ for the $3\times3$ multiplication tensor and
$R_{\F}(\T)$ for its exact tensor rank over the two-element field $\F$.
This rank is the least number of products of a linear form in the first
input with a linear form in the second, with linear operations otherwise
free~\cite{bcs}. We prove the following statement without imposing symmetry or
support conditions on the algorithm.

\begin{theorem}\label{thm:main}
Every exact bilinear algorithm for $3\times3$ matrix multiplication over
$\F$ uses at least $21$ scalar multiplications. Equivalently,
\begin{equation}\label{eq:main-bound}
 R_{\F}(\T)\geq21.
\end{equation}
\end{theorem}

Together with the explicit $23$-term decomposition in
Appendix~\ref{app:upper}, this gives
$21\leq R_{\F}(\T)\leq23$.
Wang~\cite{wang21} obtained the same numerical lower bound by
capacity-constrained profile enumeration. The proof here identifies a
different obstruction: finite geometry forces a rank equality, and
that equality supplies cross-factor relations which contradict the
forced first-factor configuration.

\subsection{The equality mechanism}
Consider a decomposition
$\T=\sum_{t=1}^{r}A_t\otimes B_t\otimes C_t$, with the factors
represented by $3\times3$ coefficient matrices. For a subspace
$W\leq M_3(\F)$, quotienting the first tensor factor by $W$ removes
every term with $A_t\in W$. A lower bound for the quotient rank
therefore gives an upper bound on the number of such terms. We call
these \emph{occupation bounds}.

A second constraint retains the row and column indices of $A_t$.
Splitting these indices gives a $27\times27$ matrix $P$ for the
multiplication tensor and matrices $M_t$ for its summands. The matrix
$P$ is a permutation matrix, while $\rank M_t=\rank A_t$ for
nonzero factors. Hence
\begin{equation}\label{eq:intro-budget}
 P=\sum_tM_t,\qquad 27\leq\sum_t\rank A_t.
\end{equation}
This lower bound on the sum of matrix ranks is not by itself a lower
bound of $21$ on the number of terms. Its role is to combine with
occupation geometry.

At length $20$, eight strengthened quotient bounds force the
first factors of rank at least two to have pairwise rank-one
differences. Such matrices lie in a common affine row or column coset.
The coset is an affine $3$-space over $\F$, where further occupation
bounds allow at most three selected points in every affine plane.
These constraints, together with~\eqref{eq:intro-budget}, force the
matrix-rank profile
\begin{equation}\label{eq:intro-profile}
 (n_1,n_2,n_3)=(16,1,3),\qquad
 \sum_t\rank A_t=16+2+9=27.
\end{equation}

The last equality is decisive. The images of the $M_t$ span the
$27$-dimensional space, and their dimensions sum to $27$; they
therefore form a direct sum. Uniqueness of its components implies
\begin{equation}\label{eq:intro-saturation}
 M_tP^{-1}M_s=\delta_{ts}M_t.
\end{equation}
For the multiplication tensor, the product on the left can be computed
explicitly:
\begin{equation}\label{eq:intro-product}
 M(A,B,C)P^{-1}M(A',B',C')
   =M(AB'C^TA',B,C').
\end{equation}
The diagonal identities show that an invertible $A_t$ requires
invertible $B_t$ and $C_t$. An off-diagonal identity would then make a
product of invertible matrices zero if two first factors were
invertible. The required three in~\eqref{eq:intro-profile} are
impossible.

The resulting obstruction is useful separately from the finite
classification: in \emph{any} decomposition whose first-factor matrix
ranks sum to $27$, at most one first factor is invertible
(Proposition~\ref{prop:invertibility} and
Corollary~\ref{cor:profiles}). For a saturated $22$-term
decomposition with nonzero first factors, this leaves only the profiles
$(17,5,0)$ and $(18,3,1)$. Thus the equality argument gives both the
final contradiction at length $20$ and restrictions on possible shorter
algorithms beyond the known upper bound.

\subsection{Background and related methods}
Bl\"aser's substitution argument gives the classical lower bound
$19$ for $3\times3$ multiplication~\cite{blaser}.
On the constructive side, SAT methods have produced many inequivalent
$23$-term algorithms~\cite{heule}, while reinforcement learning has
extended the search for exact decompositions across matrix
formats~\cite{alphatensor}. These developments make small
matrix-multiplication tensors a meeting point of algebraic complexity,
computational search, and automated mathematical reasoning.

Wang's automated finite-field framework~\cite{wang20} raised the
binary lower bound to $20$ using symmetry classes of restrictions and
recursive lower-bound certificates. We use its earlier,
$496$-representative quotient catalogue with the numeric labels fixed.
The capacity viewpoint also appears in D'Ambrosio's determination of
the binary rank of $\langle2,3,4\rangle$~\cite{dambrosio}.

Wang's proof of $21$~\cite{wang21} and ours share the occupation
bounds and the affine rank-one coset reduction. His argument
strengthens a restriction table by rank-one-span searches and
backtracking, then excludes rank-one completions of the remaining
high-rank configurations by exhaustive enumeration. Here integer
occupation certificates establish eight additional quotient bounds;
affine-plane capacities force~\eqref{eq:intro-profile}, and
\eqref{eq:intro-saturation}--\eqref{eq:intro-product} give an algebraic
contradiction involving all three factor lists.
A public review draft in Tahir's repository~\cite{tahir} gives another
counting approach, using singular first factors and
branch-and-bound exclusions of seven normal forms.

The row/column-coset lemma is a special case of the matrix-adjacency
geometry initiated by Hua~\cite{hua}; we give an elementary proof.
The rank-additivity lemma is proved over an arbitrary field.
Their combination with the finite quotient bounds is what converts
first-factor occupation into the cross-factor obstruction for $\T$.

\subsection{Formal proof and research development}
The accompanying Lean development proves the tensor conventions,
the occupation and geometric arguments, the saturation identities, and
the finite lower bounds used by the proof. Its final theorem has no
assumed finite-bound premises. The development also proves the full
$496$-representative quotient catalogue and its transport to arbitrary
subspaces. The finite computations enter through exact integer branch
certificates whose leaves and exhaustive splits are checked within Lean.

Qiushi Engine developed the structural argument during sustained
autonomous research on exact matrix multiplication.
The mathematical progression was from shorter-algorithm searches to
quotient cores, from admissible first-factor supports to the question
of factor compatibility, and finally to equality in the full tensor's
rank bound. Section~\ref{sec:development} explains these transitions
and their role in the discovery. The accompanying research
trajectory~\cite{qiushi-report} connects them to the supporting notes,
programs, counterexamples, and results, preserving intermediate
knowledge as well as the final proof.

The exposition follows the proof. Section~\ref{sec:preliminaries} fixes
the tensor and quotient conventions; Section~\ref{sec:finite} states
and certifies the finite premises. Sections~\ref{sec:geometry}
and~\ref{sec:saturation} establish the profile and its contradiction.
The remaining sections give consequences, the formalization, and the
research development. The appendices specify the finite data, a
$23$-term witness, and the formal interfaces and verification record.

\section{Tensor conventions and occupation}\label{sec:preliminaries}

Let $V=\Mat$, equipped with the nondegenerate coordinate pairing
\[
 \langle A,X\rangle=\sum_{i,j=0}^2 A_{ij}X_{ij}=\tr(A^TX).
\]
All matrix and vector coordinates are indexed by $0,1,2$. Write
$e_0=(1,0,0)^T$, $e_1=(0,1,0)^T$, $e_2=(0,0,1)^T$ and
$E_{ij}=e_i e_j^T$. Throughout the paper,
\begin{equation}\label{eq:tensor}
 \T=\sum_{i,j,k=0}^2 E_{ij}\otimes E_{jk}\otimes E_{ik}
   \in V\otimes V\otimes V.
\end{equation}
Thus its coefficient at $(3i+j,3j+k,3i+k)$ is $1$, and all other
coefficients are zero. Its contraction against $X,Y,Z$ is
$\tr(XYZ^T)$.

\begin{definition}
A decomposition of length $r$ is an identity
\begin{equation}\label{eq:decomp}
 \T=\sum_{t=1}^r A_t\otimes B_t\otimes C_t,
 \qquad A_t,B_t,C_t\in V.
\end{equation}
The tensor rank $R_{\F}(\T)$ is the least such $r$. Unless explicitly
excluded, zero summands and repeated factors are permitted in a
length-$r$ expression. A minimal expression has no zero summands.
\end{definition}

By contraction of~\eqref{eq:decomp}, the corresponding algorithm is
\begin{equation}\label{eq:algorithm}
 XY=\sum_{t=1}^r\langle A_t,X\rangle
                      \langle B_t,Y\rangle C_t
 \qquad(X,Y\in V).
\end{equation}
Conversely, evaluating~\eqref{eq:algorithm} on pairs of matrix units
recovers every coefficient of~\eqref{eq:decomp}. Matrix rank
$\rank(A_t)$ and tensor rank $R_{\F}(\T)$ will always be distinguished.

For a linear subspace $W\leq V$, put
\[
 q_W:V\longrightarrow V/W,\qquad
 T_W=(q_W\otimes\id\otimes\id)\T.
\]
The dual of $V/W$ is naturally $W^\perp$ under the coordinate pairing.
Consequently $T_W$ represents multiplication with its first input
restricted to $W^\perp$, not to $W$. Indeed, a linear functional on
$V/W$ pulls back to a functional vanishing on $W$, and every such
functional factors uniquely through $q_W$. Applying this identification
in the first tensor slot proves the assertion, including equality of
the two tensor ranks.

\begin{lemma}[Occupation inequality]\label{lem:occupation}
For a length-$r$ decomposition, let
$m(W)=|\{t:A_t\in W\}|$. Then
\begin{equation}\label{eq:occupation}
 m(W)\leq r-R_{\F}(T_W).
\end{equation}
More generally, if $W\leq U$ and $T_W$ has a length-$r$
decomposition, at most $r-R_{\F}(T_U)$ first factors, counted with
multiplicity in $V/W$, belong to $U/W$.
\end{lemma}
\begin{proof}
Applying $q_W$ to~\eqref{eq:decomp} kills every indexed term
with first factor in $W$. Discarding these terms gives a decomposition
of $T_W$ with at most $r-m(W)$ terms. For the
second assertion apply the induced map $V/W\to V/U$ instead.
\end{proof}

\subsection{Tensor symmetries}
For $U,V_0,H\in\GL_3(\F)$ the coefficient transformation
\begin{equation}\label{eq:symmetry}
 (A,B,C)\longmapsto
 (U^TAV_0^{-T},\ V_0^TBH^{-T},\ U^{-1}CH)
\end{equation}
preserves~\eqref{eq:tensor}. To check the contragredient factors,
consider the input transformation
\[
 (X,Y,Z)\longmapsto
 (UXV_0^{-1},\ V_0YH^{-1},\ U^{-T}ZH^T).
\]
It preserves $\tr(XYZ^T)$ by cancellation and cyclicity of trace.
Pullback under the coordinate pairing gives~\eqref{eq:symmetry}.
The transformation
\begin{equation}\label{eq:transpose}
 (A,B,C)\longmapsto(A^T,C,B)
\end{equation}
is also a symmetry, as follows directly by exchanging $i$ and $j$ in
the transformed sum~\eqref{eq:tensor}. These symmetries preserve the
matrix ranks of the factors and quotient tensor ranks.

The induced group on $V$ consists of $A\mapsto PAQ$ and
$A\mapsto PA^TQ$, for invertible $P,Q$. We denote it by $G$.
It has $2\cdot168^2=56{,}448$ elements. Its action on subspaces is the
one used in the frozen catalogue. Each of these first-factor actions
extends to a symmetry of the complete tensor by~\eqref{eq:symmetry}
and~\eqref{eq:transpose}.

\section{The finite quotient bounds}\label{sec:finite}

The structural proof needs three families of quotient bounds: $19$
for every line, $19$ for every two-plane containing no nonzero rank-one
matrix, and $17$ for the affine-plane spans specified below. Their
role is to limit how many first factors a short decomposition can
place in those spaces. We first identify the spaces and then explain
how their bounds are built from contractions, substitution, and
finite occupation certificates.

\subsection{The quotient catalogue}
Let $(W_i,\ell_i)_{i=0}^{495}$ be the catalogue in
Appendix~\ref{app:catalogue}. Define
\begin{equation}\label{eq:L0}
 L_0(W)=\max\bigl(\{\ell_i:W=gW_i\text{ for some }g\in G\}
                    \cup\{0\}\bigr).
\end{equation}
The catalogue is the frozen table inherited from~\cite{wang20}, not
Wang's later strengthened table. The maximum makes the definition
unambiguous even before orbit coverage and label consistency are proved.
It defines a numeric function on actual subspaces.

\begin{proposition}[Finite classification and direct bounds]\label{prop:frozen}
The catalogue covers every subspace of $V$ under $G$, and its labels
agree on overlapping orbits. The following values and bounds hold:
\begin{enumerate}
\item $L_0(0)=20$. For every nonzero $a\in V$,
$L_0(\langle a\rangle)=19$ and $R_{\F}(T_{\langle a\rangle})\geq19$.
\item The two-dimensional subspaces of $V$ have fourteen $G$-orbits.
Exactly eight consist of planes whose three nonzero matrices have
rank at least two. Their representatives are $W_{484},\ldots,W_{491}$
in Table~\ref{tab:planes}, and $L_0(W_i)=18$ for these indices.
\item Put $p=\operatorname{diag}(0,1,1)$,
$R=\{e_0w^T:w\in\F^3\}$, and $S=\langle p,R\rangle$.
Then $L_0(S)=14$ and $L_0(R)=15$. Every three-dimensional
subspace $K\leq S$ other than $R$ has $L_0(K)=17$ and
$R_{\F}(T_K)\geq17$.
\end{enumerate}
\end{proposition}

\begin{proof}[Computer-assisted proof]
The orbit coverage, label consistency, exact numeric values, and
two-plane classification are checked against the finite definition
\eqref{eq:L0}. Separately, closed quotient-rank proofs establish the
three line bounds and the normalized affine-hyperplane bounds.
Exact tensor symmetries transport them to the families stated here.
The lower-bound rules and representative source proofs are given in
Section~\ref{sec:source-bounds} and Appendix~\ref{app:source-substitution}.
These proofs precede Theorem~\ref{thm:main} in the dependency graph.
Appendix~\ref{app:finite-data} specifies the small spaces, and
Section~\ref{sec:formalization} identifies the formal interfaces.
\end{proof}

The labels specify lower bounds rather than exact quotient ranks; the
values $18$ in (ii) will be strengthened while $L_0$ remains fixed.
The line bounds already imply the full-tensor bound $20$: in any
minimal decomposition select a nonzero first factor $a$.
Lemma~\ref{lem:occupation} and the line bound give $1\leq r-19$,
hence $r\geq20$.

\subsection{Deriving the source bounds}\label{sec:source-bounds}
A source bound is a previously proved inequality for a quotient
$T_U$, used to constrain a decomposition of $T_W$ with $W\leq U$.
The following rules establish these inequalities recursively.

For $L\in V$, write $H_L=\{X:\langle L,X\rangle=0\}$.
If $W\leq H_L$, contraction by $L$ is well defined on $V/W$.
Applied to $T_W$, it gives the matrix
\begin{equation}\label{eq:source-contraction}
 C_L[(j,k),(i,l)]=L_{ij}\delta_{kl},\qquad
 \rank C_L=3\rank L.
\end{equation}
Indeed, $C_L=L^T\otimes I_3$ in the displayed ordering. Each
simple tensor contracts to a matrix of rank at most one, so
\begin{equation}\label{eq:contraction-bound}
 R_{\F}(T_W)\geq3\rank L.
\end{equation}
This gives basic bounds $3$, $6$, and $9$. Several annihilating
functionals can also be assembled into an ordinary flattening of
$T_W$. A nonsingular minor of order $d$ gives a lower bound $d$;
the finite proofs specify a minor and check a left inverse over $\F$.

Occupation can strengthen these initial bounds without finding a
larger minor. For example, let
\begin{equation}\label{eq:source-example}
 F=\begin{pmatrix}0&0&0\\0&1&0\\0&0&1\end{pmatrix},\qquad
 G_0=\begin{pmatrix}0&0&0\\0&0&1\\0&1&1\end{pmatrix},\qquad
 W=H_F\cap H_{G_0}.
\end{equation}
The matrices $F,G_0,F+G_0$ all have rank two. Their three
hyperplanes therefore have quotient rank at least $6$. Every $X\in V$
belongs to at least one of them: among the binary values
$\langle F,X\rangle$, $\langle G_0,X\rangle$, and their sum,
one is zero. In a length-$r$ decomposition of $T_W$, let $m_1,m_2,m_3$
count first factors in these hyperplanes modulo $W$. By coverage and
Lemma~\ref{lem:occupation},
\begin{equation}\label{eq:three-hyperplane-bound}
 r\leq m_1+m_2+m_3\leq3(r-6),
 \qquad\text{hence}\qquad r\geq9.
\end{equation}
This proves the catalogue bound for $W_8$: the annihilator of its
displayed basis is precisely $\langle F,G_0\rangle$. The example
shows why several small contraction bounds can give more information
than any one of them separately.

Some source bounds require substitution, a standard method in
bilinear complexity~\cite{blaser}. Regard a quotient tensor
as a family of second-factor slices $S_b$. A nonzero slice forces
some summand to have a nonzero coefficient there. Subtracting suitable
multiples of that slice from the others removes this summand, at the
cost of modifying the remaining slices. If $k$ selected slices are
independent and a flattening of rank at least $f$ survives all these
modifications, every decomposition has length at least $k+f$.
Appendix~\ref{app:source-substitution} gives the precise survival
condition and a complete $6+6=12$ example used among the source
bounds. The condition is essential: the residual after substitution,
rather than the original tensor, must retain the lower bound.

The recursive occupation step generalizes~\eqref{eq:three-hyperplane-bound}.
Suppose $W\leq U_v$, $R_{\F}(T_{U_v})\geq\lambda_v$, and
nonnegative weights $y_v$ cover each quotient class at least $c>0$
times:
\begin{equation}\label{eq:weighted-cover}
 \sum_{v:X\in U_v}y_v\geq c\quad(X\in V).
\end{equation}
Summing the occupation inequalities then gives
\begin{equation}\label{eq:weighted-bound}
 cr\leq\sum_v y_v m(U_v/W)
       \leq\sum_v y_v(r-\lambda_v).
\end{equation}
Thus a cover and previously established source bounds can exclude a
proposed length. More general exclusions split into exhaustive integer
branches before taking such nonnegative combinations; these are the
certificates described next.

For a concrete source used in the raise for $W_{484}$, consider
\[
 U=W_{450}=\Span(\operatorname{mat}(68),\operatorname{mat}(19),
 \operatorname{mat}(10)),
\]
with matrix codes as in
Appendix~\ref{app:finite-data}. Its certificate uses ten superspaces
with proved bound $17$ and eight with bound $9$. The first ten have
weights $4,4,2,2,2,2,2,2,2,2$; the others have weight one.
They give a $4$-fold cover of $V/U$, checked on its $64$ classes.
Equation~\eqref{eq:weighted-bound} therefore reads
\[
 4r\leq32r-(24\cdot17+8\cdot9)=32r-480.
\]
No $r\leq17$ satisfies this inequality, so $R_{\F}(T_U)\geq18$.
Since $W_{484}\leq U$, this is one of the bounds that forces
occupation to vanish when a length-$18$ decomposition of
$T_{W_{484}}$ is considered. The bound-$17$ sources are obtained by
the same contraction, substitution, and occupation rules.

Finally, source statements pass between equivalent representatives
using the tensor symmetries of Section~\ref{sec:preliminaries}.
Containment has a fixed direction:
\begin{equation}\label{eq:quotient-monotonicity}
 W\leq gU\quad\Longrightarrow\quad
 R_{\F}(T_W)\geq R_{\F}(T_{gU})=R_{\F}(T_U)
 \qquad(g\in G).
\end{equation}
The implication follows by applying $V/W\to V/gU$ to a
decomposition.

These rules give concrete routes to the two remaining families in
Proposition~\ref{prop:frozen}. For the line $W=\langle E_{00}\rangle$,
the $255$ planes containing $W$ partition the nonzero classes of $V/W$.
Fifteen of these planes have proved bound $17$, from the types
$W_{478},W_{479}$; the other $240$ have proved bound $18$, from
$W_{480},\ldots,W_{483}$. The zero class belongs to every plane, so
the same cover applies to decompositions with zero first factors.
Equation~\eqref{eq:weighted-bound} gives
\begin{equation}\label{eq:line-cover}
 r\leq15(r-17)+240(r-18)=255r-4575,
 \qquad\text{hence}\qquad r\geq19.
\end{equation}
The rank-two and rank-three line bounds follow by containment and
symmetry from the raises for $W_{484}$ and $W_{488}$ proved below.

For the fourteen affine-plane spans, three representatives suffice:
the coded bases $(272,4,2)$, $(273,4,2)$, and $(272,4,1)$ in
Appendix~\ref{app:finite-data}. The first is contained in a symmetry
image of $W_{262}=(132,12,2,1)$, whose bound $17$ transfers
by~\eqref{eq:quotient-monotonicity}. The other two are symmetry-equivalent
to $W_{415}$ and $W_{416}$. Their integer occupation certificates exclude length $16$
using respectively $101$ and $373$ source rows on the $63$ nonzero
quotient classes, with the branch rule described next. The remaining
eleven spans are symmetry images of the third representative.
Appendix~\ref{app:source-interfaces} locates these proofs and their
source-bound dependencies.

\subsection{The eight occupation systems}
For $i=484,\ldots,491$, put $Q_i=V/W_i$. Index nonnegative
integer weights by its nonzero classes $Q_i^*=Q_i\setminus\{0\}$,
and impose
\begin{equation}\label{eq:integer-system}
 \begin{aligned}
 x_q&\in\mathbb Z_{\geq0} &&(q\in Q_i^*),\\
 \sum_{q\in Q_i^*}x_q&=18,\\
 \sum_{q\in(U/W_i)\setminus\{0\}}x_q
       &\leq18-L_0(U) &&(W_i<U<V).
 \end{aligned}
\end{equation}
Every quotient has dimension seven, so this specifies $127$ variables
and $29{,}210$ proper nonzero quotient-subspace rows. The definition of
$L_0$ and the complete catalogue make every coefficient and bound in
these systems explicit.

\begin{proposition}[Finite occupation exclusions]\label{prop:integer}
For each $i=484,\ldots,491$, system~\eqref{eq:integer-system}
has no integer solution. Its retained source rows have proved
quotient bounds $R_{\F}(T_{U_v})\geq\lambda_v$, and its
zero-forcing rows have proved bounds at least $18$.
\end{proposition}

\begin{proof}[Integer-certificate proof]
Let $J\subseteq Q_i^*$ be the quotient directions retained after
zero-forcing, and let $U_v$ be the retained superspaces. A row
with $L_0(U)\geq18$ forces all its nonnegative summands to zero. After
these eliminations, the retained variables $z_j$ satisfy
\begin{equation}\label{eq:retained}
 z_j\geq0\ (j\in J),\qquad \sum_{j\in J}z_j\geq18,\qquad
 \sum_j a_{vj}z_j\leq18-\lambda_v,
\end{equation}
where $a_{vj}=1$ exactly when $j\in U_v/W_i$, and is zero otherwise.
The proved row label satisfies
$\lambda_v\leq L_0(U_v)$. Replacing equality of the total by a lower
bound weakens the system and is sufficient for the certificates.
Table~\ref{tab:branches} records the retained sizes.

At an internal vertex the integer branch tree splits on
$z_j\leq k$ or $z_j\geq k+1$. These alternatives are exhaustive over
$\mathbb Z$. At a leaf, take the retained inequalities, the
nonnegativity and total inequalities as needed, and the branch
inequalities along its path, all written as $a_v\cdot z\leq b_v$.
The leaf supplies nonnegative integer multipliers $y_v$ with
\begin{equation}\label{eq:farkas}
 \sum_v y_v a_v=0,\qquad \sum_v y_v b_v<0.
\end{equation}
Multiplying and adding the inequalities would give
$0\leq\sum_v y_vb_v<0$. Thus no leaf admits a solution; induction up
the exhaustive branch tree excludes every integer solution of the
retained system.

In the eight concrete certificates, Lean checks each coefficient
cancellation and strict negative right-hand side in~\eqref{eq:farkas},
and combines the leaves through explicit integer case splits. For all
$5{,}917$ source and zero-forcing rows, the data specify a basis of
$U_v$, verify $W_i<U_v<V$, and check the membership coefficients
and the bound $\lambda_v\leq L_0(U_v)$. Hence any
solution of the full system would restrict to a solution excluded by
the corresponding branch tree. The resulting eight statements quantify
over arbitrary nonnegative integer weights on the actual quotient;
they assume neither tensor realizability nor distinct support.

For the rank interpretation, every selected source bound
$R_{\F}(T_{U_v})\geq\lambda_v$ and zero-forcing bound at least $18$
has its own closed proof, connected by checked quotient transports.
Applying Lemma~\ref{lem:occupation} to these proved quotients shows
that a hypothetical short decomposition satisfies the retained
system. Numeric comparisons with $L_0$ establish the full-system
exclusion; the source proofs establish its tensor-rank interpretation.
\end{proof}

\begin{corollary}[Two-dimensional raises]\label{cor:raises}
If $W\leq V$ is two-dimensional and all three nonzero matrices in $W$
have rank at least two, then $R_{\F}(T_W)\geq19$.
\end{corollary}
\begin{proof}
First take $W=W_i$ and suppose $T_W$ has a decomposition of length
$r\leq18$. There is at least one zero-forcing superspace $D>W$
with a directly proved bound $R_{\F}(T_D)\geq18$. Quotienting to
$V/D$ already excludes $r<18$. If $r=18$, occupation in each
zero-forcing superspace is at most zero. Hence every excluded direction
and the zero quotient class have multiplicity zero. The retained
nonnegative multiplicities have total $18$, and the proved source
bounds give every row in~\eqref{eq:retained} by
Lemma~\ref{lem:occupation}. The checked branch tree excludes this
vector. Finally, classification and tensor symmetries transfer the
bound to every all-high plane.
\end{proof}

\subsection{The Boolean and integer realizations}
The research report certified the eight exclusions by CNF/DRAT
refutations checked with DRAT-trim~\cite{drat-trim}, using the earlier
certified quotient table~\cite{qiushi-report}.
The Lean proof establishes the needed source bounds within Lean
and uses integer branch certificates.
Both routes apply occupation to the same quotient tensors and prove
the same eight raises; their proof objects and dependency orders differ.
The main theorem uses the selected source bounds and integer
certificates above; the aggregate catalogue theorem is established
afterward in Proposition~\ref{prop:global}. Appendix~\ref{app:certificate-realizations}
describes the Boolean encoding and its relation to the formal proof.

\section{From occupation to the profile \texorpdfstring{$(16,1,3)$}{(16,1,3)}}\label{sec:geometry}

Assume there is a decomposition of length at most $20$. By the line
bounds and the observation after Proposition~\ref{prop:frozen}, a
minimal such decomposition has length exactly $20$ and every factor
in every term is nonzero. Fix this decomposition for the section.

\subsection{Split flattening and distinct first factors}
Define the $A$-split flattening by
\[
 \Phi_A(T)[(i,b),(j,c)]=T[3i+j,b,c],
 \quad 0\leq i,j<3,\quad 0\leq b,c<9.
\]
It is a $27\times27$ matrix. For $P=\Phi_A(\T)$, a row
$(i,3u+v)$ has its unique nonzero entry in column $(u,3i+v)$.
This correspondence is a bijection, so $P$ is a permutation matrix
of rank $27$. For a simple tensor, write
\begin{equation}\label{eq:M}
 M(A,B,C)=\Phi_A(A\otimes B\otimes C)
          =A\otimes\bigl(\operatorname{vec}(B)
                           \operatorname{vec}(C)^T\bigr).
\end{equation}
Here the Kronecker product on the right is a matrix product construction,
and vectorization is row-major. Since $B,C$ are nonzero, their outer
product has rank one and $\rank M(A,B,C)=\rank A$. Therefore
\begin{equation}\label{eq:budget}
 27=\rank P\leq\sum_{t=1}^{20}\rank A_t,
 \qquad \sum_{t=1}^{20}(\rank A_t-1)\geq7.
\end{equation}
The first inequality also holds with zero factors, using
$\rank M(A,B,C)\leq\rank A$.

The line cap is $20-19=1$. Since a line over $\F$ has just one
nonzero element, Lemma~\ref{lem:occupation} implies that the twenty
$A_t$ are pairwise distinct. Let
\[
 \mathcal H=\{A_t:\rank A_t\geq2\}.
\]
Every term contributes at most two to the excess in~\eqref{eq:budget},
so $|\mathcal H|\geq4$.

\begin{lemma}\label{lem:pair}
Distinct $a,b\in\mathcal H$ satisfy $\rank(a-b)=1$.
\end{lemma}
\begin{proof}
The span $W=\langle a,b\rangle$ is a plane containing two first
factors. If $a+b$ also had rank at least two,
Corollary~\ref{cor:raises} would give $m(W)\leq20-19=1$, a
contradiction. Since $a+b\ne0$ and subtraction equals addition in
$\F$, its rank must be one.
\end{proof}

\subsection{Affine matrix geometry}
\begin{lemma}[Affine row or column coset]\label{lem:coset}
Let $\mathcal H\subseteq\Mat$ contain at least four distinct matrices,
and suppose $\rank(a-b)=1$ for all distinct $a,b\in\mathcal H$.
Then either
\[
 \mathcal H\subseteq p+\{uw^T:w\in\F^3\}
 \quad\text{or}\quad
 \mathcal H\subseteq p+\{wv^T:w\in\F^3\}
\]
for a matrix $p$ and a nonzero vector $u$ or $v$.
\end{lemma}
\begin{proof}
Two nonzero rank-one matrices $u_1v_1^T,u_2v_2^T$ have a rank-one
sum only if $u_1=u_2$ or $v_1=v_2$. Otherwise each pair is linearly
independent over $\F$, and
\[
 u_1v_1^T+u_2v_2^T
  =\begin{pmatrix}u_1&u_2\end{pmatrix}
    \begin{pmatrix}v_1^T\\v_2^T\end{pmatrix}
\]
has rank two: the second map is surjective onto $\F^2$ and the first
is injective from $\F^2$.

Fix distinct $a,b,c\in\mathcal H$ and write
$a+b=u_1v_1^T$, $a+c=u_2v_2^T$. Their sum $b+c$ has rank one,
so they share a left or right factor. They cannot share both, since
$b\ne c$. Suppose they share the left factor $u$, with $v_1\ne v_2$.
For any other $d\in\mathcal H$, write $a+d=u'v'^T$. The rank-one
conditions for $b+d$ and $c+d$ force $u'=u$: otherwise they would
require both $v'=v_1$ and $v'=v_2$. Thus every difference from $a$
has left factor $u$, proving the column-coset alternative. The case of
a shared right factor is identical after transposition.
\end{proof}

Apply the lemma to the high-rank factors. The symmetry
\eqref{eq:transpose} exchanges the two alternatives, so it suffices to
consider a column coset. Left multiplication sends its common nonzero
column to $e_0$. Subtracting a first-row matrix from the base point
does not change the coset, so its elements can be written
\begin{equation}\label{eq:coset-block}
 \begin{pmatrix}w^T\\D\end{pmatrix},
 \qquad w\in\F^3,\quad D\in\F^{2\times3}.
\end{equation}

If $\rank D=0$, no element has rank at least two. If $\rank D=1$,
exactly six first rows lie outside its one-dimensional row space.
Those six matrices have rank two; the other two have rank one.
Distinctness of the first factors would then bound the entire rank
excess by six, contradicting~\eqref{eq:budget}. Hence $\rank D=2$.

A right change of basis sends the row space of $D$ to
$\langle e_1^T,e_2^T\rangle$, and a left action
$\operatorname{diag}(1,K)$, $K\in\GL_2(\F)$, puts the two lower
rows into that ordered basis. Both transformations preserve the
first-row direction and extend to tensor symmetries. We may therefore
assume
\begin{equation}\label{eq:normalized-coset}
 \mathcal H\subseteq p+R,\qquad
 p=\begin{pmatrix}0&0&0\\0&1&0\\0&0&1\end{pmatrix},\qquad
 R=\{e_0w^T:w\in\F^3\}.
\end{equation}
The eight labels $w$ identify this coset with affine $3$-space over
$\F$. Four matrices have rank two and four have rank three; the latter
labels are precisely $w_0=1$.

\subsection{Affine-plane capacities}
Every affine plane $\pi=w_*+D_0\subset\F^3$, with $\dim D_0=2$,
corresponds to the linear three-space
\[
 K_\pi=\Span\{p+e_0w^T:w\in\pi\}\leq S=\langle p,R\rangle.
\]
Indeed $K_\pi=\langle p+e_0w_*^T,\ e_0D_0^T\rangle$.
The first generator is outside $R$, so its dimension is three and
$K_\pi\cap(p+R)$ is exactly the four indicated coset elements.
Conversely, each three-space in $S$ other than $R$ meets $p+R$ in
such an affine plane. There are $7\cdot2=14$ of them.
Proposition~\ref{prop:frozen}(iii) and occupation give
\begin{equation}\label{eq:plane-cap}
 |\mathcal H\cap K_\pi|\leq m(K_\pi)\leq3.
\end{equation}

\begin{lemma}\label{lem:five}
Every five-element subset of $\F^3$ contains a four-point affine plane.
\end{lemma}
\begin{proof}
The ten unordered pair sums take only seven possible nonzero values.
Two distinct pairs therefore have the same sum. They cannot share a
point, since that would make their other endpoints equal. Their four
distinct endpoints have sum zero. Translating one endpoint to zero
leaves $0,u,v,u+v$ with distinct nonzero $u,v$, the points of an
affine plane.
\end{proof}

\begin{proposition}[Forced profile]\label{prop:profile}
In a hypothetical minimal length-$20$ decomposition, the counts of first
factors of matrix ranks one, two, and three are respectively
\begin{equation}\label{eq:profile}
 (n_1,n_2,n_3)=(16,1,3),\qquad
 \sum_t\rank A_t=27.
\end{equation}
\end{proposition}
\begin{proof}
Equations~\eqref{eq:plane-cap} and Lemma~\ref{lem:five} imply
$|\mathcal H|\leq4$. The budget gives the reverse inequality.
The four rank-three labels $w_0=1$ themselves form an affine plane,
so at most three of the selected factors have rank three. With four
high-rank factors the excess in~\eqref{eq:budget} is $4+n_3$.
It is at least seven, forcing $n_3=3$, $n_2=1$, and $n_1=16$.
Their rank sum is $16+2+9=27$.
\end{proof}

The fourteen bounds $17$ thus turn the rank budget into an exact
equality. The next section extracts the algebraic consequences of that
equality for the complete summands.

\section{Saturation and the cross-factor contradiction}\label{sec:saturation}

The equality in Proposition~\ref{prop:profile} is the point at which
the first-factor geometry constrains the other two factors. We first
record the linear-algebra statement responsible for this transition.

\begin{lemma}[Rank-additive saturation]\label{lem:saturation}
Let $k$ be a field and let $P,M_1,\ldots,M_r$ be $n\times n$
matrices over $k$. Suppose $P$ is invertible and
\[
 P=\sum_{t=1}^r M_t,\qquad \sum_{t=1}^r\rank M_t=n.
\]
Then, for all $s,t$,
\begin{equation}\label{eq:saturation}
 M_tP^{-1}M_s=\delta_{ts}M_t.
\end{equation}
Consequently $E_t=M_tP^{-1}$ satisfy
$\sum_tE_t=I$, $E_t^2=E_t$, and $E_tE_s=0$ for $t\ne s$.
\end{lemma}
\begin{proof}
Let $U_t=\im M_t$. Since $P$ is surjective and is the sum of the
$M_t$, the subspaces $U_t$ span $k^n$. Their dimensions sum to $n$,
so the sum is direct. For $x\in k^n$, put $z=P^{-1}M_sx$. Then
\[
 \sum_t M_tz=Pz=M_sx\in U_s.
\]
Uniqueness of components in $\bigoplus_tU_t$ gives $M_tz=0$ if
$t\ne s$ and $M_sz=M_sx$. This is~\eqref{eq:saturation}.
Right multiplication by $P^{-1}$ gives the idempotent identities;
their sum is $PP^{-1}=I$.
\end{proof}

Orthogonality here means pairwise annihilation, not orthogonality for
an inner product. The lemma allows zero summands and arbitrary fields.

\begin{lemma}[Product formula]\label{lem:product}
For the convention~\eqref{eq:tensor}, and arbitrary $A,B,C,A',B',C'
\in\Mat$,
\begin{equation}\label{eq:product}
 M(A,B,C)P^{-1}M(A',B',C')
   =M(AB'C^TA',B,C').
\end{equation}
\end{lemma}
\begin{proof}
Since $P$ is a permutation matrix, $P^{-1}=P^T$. For composite
indices as in Section~\ref{sec:geometry},
\[
 P^{-1}[(k,d),(m,e)]=T[3m+k,e,d].
\]
The $((i,b),(j,c))$ entry of the left-hand side
of~\eqref{eq:product} is therefore $B[b]C'[c]$ times
\[
 \sum_{k,m=0}^2 A[i,k]A'[m,j]
       \sum_{d,e=0}^{8} C[d]T[3m+k,e,d]B'[e].
\]
The inner sum has nonzero tensor coefficients only at
$e=3k+h$, $d=3m+h$, for $0\leq h<3$, so it equals
\[
 \sum_{h=0}^2 C[m,h]B'[k,h]=(B'C^T)[k,m].
\]
The full entry is thus
$B[b]C'[c](AB'C^TA')[i,j]$, as claimed.
\end{proof}

If $B$ and $C$ are nonzero, the map $X\mapsto M(X,B,C)$ is
injective. To see this, choose entries $B[b]\ne0$ and $C[c]\ne0$.
The entries with these fixed $b,c$ recover all entries of $X$ up to
the same nonzero scalar. In particular $M(X,B,C)=0$ implies $X=0$.

\begin{proposition}[Invertibility obstruction at saturation]
\label{prop:invertibility}
In a decomposition of $\T$ with nonzero factors and
$\sum_t\rank A_t=27$, at most one first factor is invertible.
\end{proposition}
\begin{proof}
Set $M_t=M(A_t,B_t,C_t)$. Their sum is $P$, and their ranks sum to
$27$. Lemmas~\ref{lem:saturation} and~\ref{lem:product}, followed by
the preceding injectivity observation, give
\begin{align}
 A_tB_tC_t^TA_t&=A_t,\label{eq:diagonal}\\
 A_tB_sC_t^TA_s&=0\quad(t\ne s).\label{eq:offdiagonal}
\end{align}
For an invertible $A_t$, multiplying~\eqref{eq:diagonal} on both
sides by $A_t^{-1}$ yields
\begin{equation}\label{eq:inverse}
 B_tC_t^T=A_t^{-1}.
\end{equation}
Thus both $B_t$ and $C_t$ are invertible. If $A_t$ and $A_s$ are
invertible at two distinct indices, all four factors in
$A_tB_sC_t^TA_s$ are invertible by~\eqref{eq:inverse}. Their product
cannot be zero, contradicting~\eqref{eq:offdiagonal}.
\end{proof}

\begin{proof}[Proof of Theorem~\ref{thm:main}]
The line quotient bounds exclude minimal decompositions shorter than
$20$. If a minimal length-$20$ decomposition existed,
Proposition~\ref{prop:profile} would give rank sum $27$ and three
invertible first factors. Proposition~\ref{prop:invertibility} permits
at most one. Hence there is no decomposition of length at most $20$,
and~\eqref{eq:algorithm} gives the asserted lower bound for bilinear
algorithms.
\end{proof}

\section{Saturated decompositions and further consequences}\label{sec:consequences}

For a length-$r$ decomposition define the $A$-slot excess by
\[
 \Delta_A=\sum_{t=1}^r\rank A_t-27.
\]
The split bound shows that $\Delta_A\geq0$. A decomposition is
\emph{saturated in the $A$-slot} if $\Delta_A=0$.

\begin{corollary}[Zero terms and saturated profiles]\label{cor:profiles}
In an arbitrary length decomposition saturated in the $A$-slot,
including one with zero summands, at most one $A_t$ is invertible.
If its length is $22$ and all $A_t$ are nonzero, its $A$-slot
matrix-rank profile is either
\begin{equation}\label{eq:22profiles}
 (17,5,0)\quad\text{or}\quad(18,3,1).
\end{equation}
Without the nonzero-$A_t$ hypothesis, either these profiles hold or
deleting a zero first factor gives a saturated length-$21$
decomposition.
\end{corollary}
\begin{proof}
For arbitrary factors,
\[
 27\leq\sum_t\rank M(A_t,B_t,C_t)
       \leq\sum_t\rank A_t=27.
\]
Hence each individual rank deficit is zero. In particular if $A_t\ne0$,
neither $B_t$ nor $C_t$ can be zero. Removing all terms with $A_t=0$
leaves nonzero factors and preserves the rank sum, so
Proposition~\ref{prop:invertibility} applies.

For $22$ nonzero first factors, write $n_j$ for the count of rank $j$.
Then $n_1+n_2+n_3=22$ and $n_1+2n_2+3n_3=27$, whence
$n_2+2n_3=5$. Since $n_3\leq1$, the only possibilities
are~\eqref{eq:22profiles}. If a first factor is zero, deletion leaves
length $21$ and the same rank sum. Two such zero factors would produce
length $20$, contrary to Theorem~\ref{thm:main}.
\end{proof}

The cyclic symmetry
\[
 (A,B,C)\longmapsto(B,C^T,A^T)
\]
follows by relabelling $i,j,k$ in~\eqref{eq:tensor}, and its third
iterate is the identity. Transpose preserves matrix rank. Thus the
split bound, the invertibility obstruction, and
Corollary~\ref{cor:profiles} hold separately in the $B$- and $C$-slots.
Saturation in one slot does not imply saturation in either other slot.

\subsection{Positive rank excess}
At length $21$, the three occupation capacities driving the
length-$20$ argument change as follows:
\begin{center}
\begin{tabular}{lcc}
\toprule
Subspace family & Length $20$ & Length $21$\\
\midrule
Lines with bound $19$ & $1$ & $2$\\
All-high planes with bound $19$ & $1$ & $2$\\
Affine-plane spans with bound $17$ & $3$ & $4$\\
\bottomrule
\end{tabular}
\end{center}
Repeated first factors and two high-rank factors in an all-high plane
are no longer immediately excluded. Four coset labels forming an affine
plane are also allowed. The argument therefore no longer forces the
profile or rank equality. At positive excess the hypotheses of
Lemma~\ref{lem:saturation} are unavailable.

The excess has a precise linear-algebraic meaning. For nonzero
summands set $U_t=\im M_t$. The addition map
\[
 \Sigma:\bigoplus_t U_t\longrightarrow\F^{27},\qquad
 (u_t)_t\longmapsto\sum_t u_t
\]
is surjective because $\sum_tM_t=P$ is invertible. Rank--nullity gives
\begin{equation}\label{eq:excess-kernel}
 \dim\ker\Sigma=\sum_t\dim U_t-27=\Delta_A.
\end{equation}
Zero excess means that components are unique, which is exactly what
produces the identities in Lemma~\ref{lem:saturation}. Positive excess
measures the remaining linear relations between these components.
Extending the method to length $21$ therefore requires stronger
occupation bounds or cross-factor relations that remain effective
when this kernel has small positive dimension.
The interval $21$--$23$ leaves both improved constructions and stronger
lower bounds to be determined.

\subsection{Global frozen-table soundness}
The full catalogue extends the selected quotient bounds to a single
lower-bound function on all subspaces of $V$.

\begin{proposition}\label{prop:global}
For every frozen representative and every subspace $W\leq V$,
\[
 R_{\F}(T_{W_i})\geq\ell_i\quad(0\leq i<496),
 \qquad R_{\F}(T_W)\geq L_0(W).
\]
\end{proposition}
\begin{proof}[Computer-assisted proof]
The completed registry supplies a closed proof at each of the $496$
listed bases and labels. Its bindings check equal spans, containment
with quotient monotonicity, or invertible tensor transport to the
individual source proofs. The entry $W_{495}=0$, labelled $20$, may
now use Theorem~\ref{thm:main}, since $T_0=\T$. For each $G$-image,
tensor symmetry preserves the quotient rank. Taking the maximum in
\eqref{eq:L0} gives the second inequality.
\end{proof}

With universal soundness available, there is also a direct link from
any short quotient decomposition to the full occupation system. At
length $r\leq18$, its nonzero first-factor multiplicities satisfy the
caps $r-L_0(U)$. Increasing any one multiplicity by $18-r$ preserves
every enlarged cap $18-L_0(U)$ and makes the total $18$. This explains
the padding argument for the full systems used in the research report.
The proof above instead extracts the retained system from separately
established source bounds, as in the Lean main-theorem dependency graph.

\section{Lean formalization}\label{sec:formalization}

The Lean~4 formalization~\cite{lean4} follows the mathematical argument from the definition
of multiplication to the lower bound. It includes the finite
computations that establish the quotient premises, as well as the
geometric and algebraic deductions. In particular, the statements connecting quotient tensors to finite
integer systems are proved together with the certificate exclusions,
so the same formal argument covers their mathematical interpretation
and their application to the main theorem.

\subsection{Mathematical statements}
The field is represented by \code{ZMod 2}, and each factor is an
actual $3\times3$ matrix. The entrywise definition
of~\eqref{eq:tensor} is proved equivalent to the algorithm
identity~\eqref{eq:algorithm} for every pair of input matrices.
The final declaration \code{bilinear\_mul\_requires\_21} states that
every such algorithm has at least $21$ products. Zero summands are
allowed in the quantified family; no minimality, distinctness, or
symmetry condition is assumed.

The tensor-rank statement is \code{QiushiMatmul.rank\_ge\_21}.
Its proof combines the structural argument with a closed collection of
finite premises: bound $19$ for the three line types and the eight
all-high plane types, and bound $17$ for the fourteen normalized
affine-plane spans. The explicit upper-bound witness then gives
\code{rank\_between\_21\_and\_23}. These declarations formalize
Theorem~\ref{thm:main} and its stated interval.

\subsection{Finite proofs and reusable components}
The integer certificates of Section~\ref{sec:finite} are expanded
into Lean proofs. Each leaf proves a contradiction by an exact
nonnegative combination of inequalities. Each internal node applies an
exhaustive integer split. Quotient membership and the source rank
bounds are proved before these certificates are applied to a
hypothetical decomposition. Thus the final theorem derives its finite
premises within the same proof system as the saturation argument.

The development uses Mathlib's finite fields, matrices, subspaces,
linear maps, dimension theory, and rank
identities~\cite{mathlib}. The problem-specific components include
quotient/annihilator equivalence, occupation extraction, symmetry
transport, affine-coset geometry, and the multiplication-specific
product formula. Rank-additive saturation is formulated over any
field. It can be reused independently of the binary finite data.
The complete frozen-table theorem proves all $496$ representative
bounds and transports them to arbitrary subspaces; its logical place
is after the main theorem, as in
Proposition~\ref{prop:global}.

\subsection{Kernel verification}
The final theorems and their full transitive proof dependencies passed
Lean's unchanged kernel replay procedure from empty
environments~\cite{lean-replay}. Their axiom dependencies are confined to
\[
 \{\code{propext},\ \code{Classical.choice},\ \code{Quot.sound}\}.
\]
The main theorem uses no additional axiom for a quotient bound and no
external solver-success assumption. Its finite proofs are checked
proof terms, not trusted outputs of the certificate generators.

Appendix~\ref{app:formalization} lists the principal interfaces, the
pinned dependencies, and the replay groups.
The accompanying formalization package~\cite{qiushi-lean} contains the
sources, certificate data, generators, and commands for reproduction.
Neither a running Qiushi Engine nor access to a language model is
needed to check the proof.

\section{Qiushi Engine and the development of the proof}\label{sec:development}

Qiushi Engine developed the structural proof through sustained
autonomous research on exact matrix multiplication. Its research
planning, construction, computation, and critical review operate with
persistent mathematical memory, as described in its earlier work on
an optical experimental platform~\cite{qiushi-optics}. Here the system
began with the search for a shorter algorithm. The eventual lower-bound
proof grew from a different question exposed by that search: which
constraints on candidate factors express compatibility with the full
multiplication tensor?

The public research trajectory~\cite{qiushi-report} brings together
the relevant notes, programs, results, and corrections. It is organized
by mathematical question rather than as a transcript or a strictly
linear chronology. The following three developments explain how the
representations changed and what was retained from each.

\subsection{From numerical dependence to an exact deletion criterion}
A general $22$-term ansatz contains $22(9+9+9)=594$ coefficients
subject to $729$ cubic tensor-coordinate identities. The early search
also considered real-coefficient deformations of known $23$-term
schemes. The aim was to reach a configuration in which one product
could be eliminated. A nearly dependent first--second factor pairing
initially appeared to offer such a route.

The difficulty can be stated exactly. For a numerical decomposition
$\mathcal D=\sum_{t=0}^{22}u_t\otimes v_t\otimes w_t$ over
$\mathbb R$, normalize the nonzero pairing columns by setting
\[
 d_t=\frac{u_t\otimes v_t}{\|u_t\otimes v_t\|_2},\qquad
 \widetilde w_t=\|u_t\otimes v_t\|_2 w_t,\qquad
 D=[d_0\ \cdots\ d_{22}].
\]
Let $\lambda$ be a unit right singular vector for
$\sigma_{\min}(D)$, and choose $k$ with $\lambda_k\ne0$.
If $D\lambda$ were zero, the $k$-th pairing column could be
eliminated by replacing each remaining output vector with
$\widetilde w_t-(\lambda_t/\lambda_k)\widetilde w_k$.
For an approximate dependence the same operation gives a $22$-term
tensor $\mathcal D_{22}$, with exact error relative to the current
tensor
\begin{equation}\label{eq:deletion-diagnostic}
 \|\mathcal D-\mathcal D_{22}\|_F
   =\frac{\|D\lambda\|_2\,\|\widetilde w_k\|_2}{|\lambda_k|}
   =\frac{\sigma_{\min}(D)\,\|\widetilde w_k\|_2}{|\lambda_k|}.
\end{equation}
Thus a small pairing singular value alone is insufficient. The output
factor may grow at the compensating rate.

To compare the result with the target tensor, write
$E=\|\mathcal D-\mathcal D_{22}\|_F$ and
$\varepsilon=\|\T-\mathcal D\|_F$, using the real tensor with the
coefficients in~\eqref{eq:tensor}. The triangle inequalities give
\begin{equation}\label{eq:deletion-target}
 |E-\varepsilon|\leq\|\T-\mathcal D_{22}\|_F
                  \leq E+\varepsilon.
\end{equation}
Table~\ref{tab:continuation} reports both quantities for a continuation
from the integer $23$-term scheme \code{serendipitous\_8d34} in
Perminov's FastMatrixMultiplication collection~\cite{perminov-code}.
The output vectors were transposed from the source's column-major
convention to~\eqref{eq:tensor}; the source version and conversion
are specified in Appendix~\ref{app:trajectory-sources}.
The parameter $t$ labels saved continuation points, not elapsed time.
The same column, labelled $10$ in the data, is eliminated at each
point. Across these samples the pairing singular value falls by two
orders of magnitude, but output growth keeps the deletion error of
order one. The target residual remains below $4\times10^{-12}$.
Thus~\eqref{eq:deletion-target} shows that the deleted representation
is still about one unit from the target at each of these points.

\begin{table}[htbp]
\centering
\caption{Pairing dependence and compensating output growth in a
recorded real-coefficient continuation. Here $E$ is the deletion error
in~\eqref{eq:deletion-diagnostic}, and $\varepsilon$ is the target
residual before deletion. Values are rounded from the archived data.}
\label{tab:continuation}
\small
\begin{tabular}{rrrrrr}
\toprule
$t$ & $\sigma_{\min}(D)$ & $|\lambda_{10}|$ & $\|\widetilde w_{10}\|_2$ & $E$ & $\varepsilon$\\
\midrule
10 & $5.265\times10^{-3}$ & $0.7084$ & $1.400\times10^2$ & $1.040$ & $1.997\times10^{-14}$\\
20 & $7.051\times10^{-4}$ & $0.7076$ & $1.064\times10^3$ & $1.061$ & $1.902\times10^{-13}$\\
60 & $4.745\times10^{-5}$ & $0.7075$ & $1.779\times10^4$ & $1.193$ & $3.645\times10^{-12}$\\
\bottomrule
\end{tabular}
\end{table}

Qiushi's correction separated projective proximity from the tensor
contribution it was meant to remove. The retained lesson was an exact
one: a change to the first two factor lists must be assessed together
with the output coefficients. The archive includes the continuation
analysis, its corrected interpretation, and the data used in
Table~\ref{tab:continuation}; the source links are collected in
Appendix~\ref{app:trajectory-sources}.

\subsection{Quotient cores and the completion problem}
The finite-field search made a coordinate quotient a natural target.
Taking $W=\langle E_{00}\rangle$ in the first factor gives an
$8\times9\times9$ core. This is the $E_{11}$ core of the research
report, whose historical name uses one-based indices; the present
paper uses indices $0,1,2$. A $19$-term decomposition of this core,
embedded in the coordinate complement, would need only the three
deleted terms
\begin{equation}\label{eq:core-restoration}
 \sum_{k=0}^2 E_{00}\otimes E_{0k}\otimes E_{0k}
\end{equation}
to recover $\T$. It would therefore give a $22$-product algorithm.
The smaller object retained a direct route to the original problem.

Occupation constraints organized candidate first factors into finite
configurations. Their feasibility, however, did not solve completion.
For fixed first factors $a_t\in V/W$, the remaining obligation is
\begin{equation}\label{eq:core-completion}
 \sum_t a_t B_t[j,k]C_t[i,l]
   =\delta_{kl}\,q_W(E_{ij})
 \qquad(0\leq i,j,k,l<3).
\end{equation}
Each pair $B_t,C_t$ must satisfy all these equations simultaneously.
Occupation records where the $a_t$ lie, but omits these products and
their shared coefficients. Fixed-factor solving, contractions, lift
conditions, and relaxations examined different parts of this gap.

The gap also appeared in verification. An early cardinality encoding
reused auxiliary variables between constraints and produced a
spurious contradiction. Qiushi identified the collision, withdrew the
affected exclusion, and rebuilt the encoding with disjoint counters.
The mathematical requirement was then explicit: every legal
multiplicity vector must extend to a simultaneous Boolean assignment.
This requirement survives in the formal copy and counter semantics
described in Appendix~\ref{app:certificate-realizations}.

The quotient search consequently left two distinct assets: the
restoration formula for upper bounds and valid occupation constraints
for lower bounds. The latter did not have to solve completion directly.
They could instead force a configuration in which the full tensor
itself supplied additional equations.

\subsection{Returning to the full tensor}
The split flattening introduced information absent from a quotient
support: each complete algorithm term contributes an image space of
dimension $\rank A_t$, and these spaces together must span dimension
$27$. After the geometric reduction, the profile $(16,1,3)$ used
exactly this dimension. There was no remaining dimensional overlap.
Qiushi then used the uniqueness of components in that direct sum to
obtain the saturation identities, and expanded their products in the
coordinates of matrix multiplication.

This was a change in the kind of conclusion being sought. Earlier
finite work had considered excluding completions of the surviving
profile by additional support systems. The saturation argument instead
gave a statement about every complete decomposition at rank sum $27$:
at most one first factor is invertible. It replaced the terminal
support exclusions by Proposition~\ref{prop:invertibility}. The
research note on the shortened proof also identified that the proposed
strengthening of orbit $479$ was unnecessary for the coset reduction:
its high-rank pair already has rank-one difference. That step requires
only the eight all-high plane raises. The existing bound
$R_{\F}(T_{W_{479}})\geq17$ still enters the line
cover~\eqref{eq:line-cover}.

The same statement then returned information to the constructive
problem. It restricts saturated $22$-term algorithms through
Corollary~\ref{cor:profiles}, while positive rank excess identifies
where the argument must be extended.

\subsection{Research memory and formal development}
Qiushi Engine's Meta-Trace links the history of questions, the evidence
for claims, and changes in the research trajectory~\cite{qiushi-optics}.
For this study, the accompanying record contains the deletion
diagnostic, the core restoration map,
the distinction between occupation and completion, the counter
correction, and the saturation argument, together with the notes,
programs, and computed data supporting them. Superseded interpretations
are retained alongside the evidence that corrected them.

The subsequent Lean development and manuscript preparation included
additional Codex assistance in proof completion, integration, and
verification. Formalization preserved the structural argument while
rebuilding its finite premises as proved quotient statements and
integer certificates. The library also includes the occupation,
substitution, and saturation lemmas used in the proof.

\appendix
\section{Explicit finite premises}\label{app:finite-data}

\subsection{Matrix codes and representative planes}
For $0\leq v<512$ define $\operatorname{mat}(v)_{ij}$ to be the bit
at position $3i+j$ of $v$. Thus the least significant bit represents
$E_{00}$, and code $272$ represents $\operatorname{diag}(0,1,1)$.
A displayed tuple of codes denotes their $\F$-linear span. The frozen
RREF convention chooses pivots from bit $8$ down to bit $0$.
The following ranks are the sorted ranks of the three nonzero elements,
not necessarily the ranks in basis order.

\begin{table}[htbp]
\centering
\caption{The eight all-high planes. Every inherited label is $18$;
every improved quotient lower bound is $19$.}
\label{tab:planes}
\begin{tabular}{rcrr}
\toprule
Index & Basis codes & Rank multiset & Orbit size\\
\midrule
484 & $(19,10)$ & $(2,2,2)$ & 98\\
485 & $(20,10)$ & $(2,2,2)$ & $2{,}352$\\
486 & $(68,10)$ & $(2,2,2)$ & $1{,}176$\\
487 & $(84,10)$ & $(2,3,3)$ & $3{,}528$\\
488 & $(96,10)$ & $(2,2,3)$ & $14{,}112$\\
489 & $(258,10)$ & $(2,2,2)$ & $4{,}704$\\
490 & $(275,10)$ & $(2,3,3)$ & $4{,}704$\\
491 & $(163,84)$ & $(3,3,3)$ & $1{,}344$\\
\bottomrule
\end{tabular}
\end{table}

Their orbit sizes sum to $32{,}018$. There are $43{,}435$ two-planes
in total. The six other orbit representatives have bases
$(2,1)$, $(10,1)$, $(16,1)$, $(20,1)$, $(84,1)$, and $(160,1)$,
with indices $478$--$483$. Each contains a rank-one matrix and is
outside the hypothesis of Corollary~\ref{cor:raises}.

The total number of subspaces is obtained by counting ordered bases:
each $d$-space has $\prod_{j=0}^{d-1}(2^d-2^j)$ ordered bases,
whereas $V$ has $\prod_{j=0}^{d-1}(2^9-2^j)$ independent
ordered $d$-tuples. Thus
\[
 \sum_{d=0}^9 \qbinom{9}{d}=8{,}283{,}458,
 \qquad
 \qbinom{n}{d}=\prod_{j=0}^{d-1}\frac{2^n-2^j}{2^d-2^j}.
\]
Similarly $\qbinom{9}{2}=43{,}435$, and
$\sum_{d=1}^6\qbinom{7}{d}=29{,}210$. These count actual subspaces;
the formal orbit classification supplies the separate coverage claim.

\subsection{The normalized hyperplanes}
The spaces $S=(272,4,2,1)$ and $R=(4,2,1)$ have dimensions four
and three. Their frozen values are $14$ and $15$. The other fourteen
hyperplanes of $S$ are listed below; each has frozen value $17$.
The index is the suffix in the formal name \code{affineHyperplane}.
\begin{center}
\begin{tabular}{rc@{\qquad}rc}
\toprule
Index & Basis codes & Index & Basis codes\\
\midrule
0 & $(272,4,2)$ & 7 & $(276,2,1)$\\
1 & $(273,4,2)$ & 8 & $(272,5,2)$\\
2 & $(272,4,1)$ & 9 & $(273,5,2)$\\
3 & $(274,4,1)$ & 10 & $(272,6,1)$\\
4 & $(272,4,3)$ & 11 & $(274,6,1)$\\
5 & $(273,4,3)$ & 12 & $(272,5,3)$\\
6 & $(272,2,1)$ & 13 & $(273,5,3)$\\
\bottomrule
\end{tabular}
\end{center}
Equivalently these are
$\{\alpha p+e_0w^T:a\cdot w=b\alpha\}$ for
$a\in\F^3\setminus0$ and $b\in\F$. This formula proves that
the list exhausts the fourteen affine-plane spans used in the proof.

The line representatives in the small premise interface are
$(1)$, $(17)$, and $(273)$, of matrix ranks one, two, and three.
They are respectively equivalent to the frozen line representatives
$(1)$, $(10)$, and $(84)$, all with label $19$.

\subsection{Integer branch-certificate sizes}
Table~\ref{tab:branches} describes the Lean integer branch certificates.
A zero-forcing row has target-level
bound $18$ and removes its quotient directions by nonnegativity.
Source rows retain their proved labels in~\eqref{eq:retained}.
All systems begin with $127$ nonzero directions and $29{,}210$ rows;
retaining fewer rows is a weakening, so infeasibility of the retained
system proves infeasibility of the full one. The coefficient tests,
row bindings, and exhaustive branches are part of the formal proof.
\begin{table}[htbp]
\centering
\caption{Integer branch certificates after zero-weight elimination.}
\label{tab:branches}
\begin{tabular}{rrrrr}
\toprule
Plane & Live variables & Source rows & Zero-forcing rows & Leaves\\
\midrule
484 & 43 & 1{,}665 & 84 & 569\\
485 & 41 & 421 & 86 & 40\\
486 & 41 & 262 & 86 & 14\\
487 & 48 & 668 & 79 & 89\\
488 & 41 & 413 & 86 & 32\\
489 & 32 & 371 & 95 & 43\\
490 & 43 & 713 & 84 & 115\\
491 & 49 & 726 & 78 & 120\\
\bottomrule
\end{tabular}
\end{table}

\subsection{Substitution and a source bound of twelve}
\label{app:source-substitution}
The substitution rule used for the source bounds has a useful
formulation in terms of slices. Let $D$ be a finite-dimensional
$\F$-vector space, let $\beta,\gamma$ be finite index sets, and write
a tensor as $S:\beta\times\gamma\to D$. Its second-factor slices
are the functions $S_b:c\mapsto S(b,c)$.

\begin{lemma}[Substitution with a surviving flattening]
\label{lem:source-substitution}
Choose $b_1,\ldots,b_k\in\beta$ such that the slices $S_{b_l}$
are linearly independent, and $c_1,\ldots,c_f\in\gamma$.
Suppose that for every nonzero $y\in\F^f$ there exist $b\in\beta$
and $\phi\in D^*$ satisfying
\begin{equation}\label{eq:substitution-witness}
 \phi\left(\sum_j y_j S(b_l,c_j)\right)=0\quad(1\leq l\leq k),
 \qquad
 \phi\left(\sum_j y_j S(b,c_j)\right)\ne0.
\end{equation}
Then $R_{\F}(S)\geq k+f$.
\end{lemma}
\begin{proof}
Consider a decomposition $S(b,c)=\sum_{t=1}^r a_t B_t(b)C_t(c)$.
If $k>0$, the nonzero slice at $b_1$ forces $B_{t_0}(b_1)=1$
for some $t_0$.
Replace the slice family by
\[
 S'(b,c)=S(b,c)+B_{t_0}(b)S(b_1,c).
\]
The $t_0$-th term cancels, leaving a decomposition with at most $r-1$
terms. The other selected slices remain independent: a relation among
them would give a relation among the original $k$ selected slices,
with the coefficient of $S_{b_1}$ adjusted accordingly.
For a witness in~\eqref{eq:substitution-witness}, the added multiple
of $S_{b_1}$ vanishes under $\phi$. Hence the same condition holds
for $S'$ and its $k-1$ remaining selected slices.

Iterate $k$ times. In the final family, no nonzero linear combination
of the $f$ selected columns vanishes, since the surviving witness
detects it at some $b$. These columns are therefore independent in
the flattening with column index $c$ and row space $D\otimes\F^\beta$.
Its rank is at least $f$, whereas a decomposition with at most $r-k$
terms gives rank at most $r-k$. Thus $r\geq k+f$.
\end{proof}

For a source actually used in the formal development, set
\[
 J=\begin{pmatrix}0&0&1\\0&1&0\\1&0&0\end{pmatrix},
 \qquad K=E_{22},\qquad W=H_J\cap H_K.
\]
This is $W_{11}$ in the catalogue. Quotient coordinates are
$X\mapsto(\langle J,X\rangle,\langle K,X\rangle)$, so its slice
family, with $b=(j,k)$ and $c=(i,l)$, is
\begin{equation}\label{eq:source-twelve-slices}
 S((j,k),(i,l))=
    (\delta_{i+j,2}\delta_{kl},\ \delta_{i,2}\delta_{j,2}\delta_{kl}).
\end{equation}
Choose the six slices with $j=0,1$ and all $k$, and the six columns
with $i=0,2$ and all $l$. The slices are independent: in the first
component, each has a different nonzero coordinate $(i,l)=(2-j,k)$.

Write a combination of the selected columns as $y=(a,b)\in\F^3\oplus\F^3$,
where $a$ corresponds to $i=0$ and $b$ to $i=2$.
If $b\ne0$, the second coordinate functional annihilates all six
selected slice evaluations and detects $b_k\ne0$ in a slice $(2,k)$.
If $b=0$ and $a\ne0$, the first coordinate functional instead
annihilates those evaluations and detects $a_k\ne0$ in a slice
$(2,k)$. This proves~\eqref{eq:substitution-witness} for every
nonzero $y$. Lemma~\ref{lem:source-substitution} yields
$R_{\F}(T_{W_{11}})\geq6+6=12$.
This example and~\eqref{eq:three-hyperplane-bound} illustrate the two
complementary operations used in the source proofs: preserving a
flattening through substitution, and combining quotient capacities.

\section{The complete frozen catalogue}\label{app:catalogue}

Table~\ref{tab:catalogue} gives every pair $(W_i,\ell_i)$ used in
definition~\eqref{eq:L0}. Parentheses enclose basis codes in the convention
of Appendix~\ref{app:finite-data}; the empty tuple is the zero space.
The dimension is the number of basis vectors. The labels are frozen
lower bounds, not asserted exact ranks. In particular entries
$484$--$491$ remain labelled $18$ here even though their quotient ranks
are at least $19$.

Together with the explicitly defined group $G$, this catalogue specifies
all eight systems~\eqref{eq:integer-system} without an external lookup
file. For any superspace row $U$, find a $G$-image of a listed basis
spanning $U$ and use its label. The proved coverage and consistency make
this procedure total and independent of the representative found.
The local source package also includes the same catalogue as structured
JSON. Its transcription is checked against all $496$ formal registry
entries and the literal Lean basis and label definitions.

\begingroup
\small
\setlength{\tabcolsep}{4pt}
\renewcommand{\arraystretch}{1.08}
\begin{longtable}{@{}rlr@{\hspace{12mm}}rlr@{}}
\caption{All $496$ frozen representative bases and labels.}\label{tab:catalogue}\\
\toprule
$i$ & Basis of $W_i$ & $\ell_i$ & $i$ & Basis of $W_i$ & $\ell_i$\\
\midrule\endfirsthead
\multicolumn{6}{c}{\tablename\ \thetable: continued}\\
\toprule
$i$ & Basis of $W_i$ & $\ell_i$ & $i$ & Basis of $W_i$ & $\ell_i$\\
\midrule\endhead
\midrule\multicolumn{6}{r}{Continued on the next page}\\\endfoot
\bottomrule\endlastfoot
0 & $(256,128,64,32,16,8,4,2,1)$ & 0 & 248 & $(8,4,2,1)$ & 14\\
1 & $(128,64,32,16,8,4,2,1)$ & 3 & 249 & $(80,4,2,1)$ & 14\\
2 & $(256,160,64,16,8,4,2,1)$ & 6 & 250 & $(16,8,2,1)$ & 15\\
3 & $(256,128,68,32,20,8,2,1)$ & 9 & 251 & $(20,8,2,1)$ & 16\\
4 & $(64,32,16,8,4,2,1)$ & 6 & 252 & $(32,8,2,1)$ & 15\\
5 & $(160,64,16,8,4,2,1)$ & 9 & 253 & $(68,8,2,1)$ & 16\\
6 & $(256,64,16,8,4,2,1)$ & 6 & 254 & $(84,8,2,1)$ & 16\\
7 & $(256,96,16,8,4,2,1)$ & 9 & 255 & $(96,8,2,1)$ & 16\\
8 & $(304,160,64,8,4,2,1)$ & 9 & 256 & $(160,8,2,1)$ & 15\\
9 & $(256,160,80,8,4,2,1)$ & 9 & 257 & $(256,8,2,1)$ & 14\\
10 & $(256,160,68,16,8,2,1)$ & 12 & 258 & $(272,8,2,1)$ & 15\\
11 & $(128,68,32,20,8,2,1)$ & 12 & 259 & $(32,12,2,1)$ & 15\\
12 & $(256,128,32,20,8,2,1)$ & 9 & 260 & $(80,12,2,1)$ & 16\\
13 & $(256,160,68,20,8,2,1)$ & 12 & 261 & $(96,12,2,1)$ & 16\\
14 & $(256,128,96,20,8,2,1)$ & 12 & 262 & $(132,12,2,1)$ & 17\\
15 & $(256,128,84,32,8,2,1)$ & 9 & 263 & $(160,12,2,1)$ & 16\\
16 & $(256,148,80,32,12,2,1)$ & 12 & 264 & $(256,12,2,1)$ & 16\\
17 & $(262,128,68,32,20,10,1)$ & 14 & 265 & $(272,12,2,1)$ & 16\\
18 & $(32,16,8,4,2,1)$ & 9 & 266 & $(80,32,2,1)$ & 15\\
19 & $(64,16,8,4,2,1)$ & 9 & 267 & $(84,32,2,1)$ & 16\\
20 & $(96,16,8,4,2,1)$ & 12 & 268 & $(256,32,2,1)$ & 14\\
21 & $(256,16,8,4,2,1)$ & 9 & 269 & $(264,32,2,1)$ & 16\\
22 & $(160,64,8,4,2,1)$ & 9 & 270 & $(320,32,2,1)$ & 15\\
23 & $(128,80,8,4,2,1)$ & 12 & 271 & $(336,32,2,1)$ & 15\\
24 & $(160,80,8,4,2,1)$ & 11 & 272 & $(152,80,2,1)$ & 15\\
25 & $(256,80,8,4,2,1)$ & 12 & 273 & $(156,80,2,1)$ & 16\\
26 & $(288,128,8,4,2,1)$ & 9 & 274 & $(160,80,2,1)$ & 16\\
27 & $(304,160,8,4,2,1)$ & 11 & 275 & $(160,84,2,1)$ & 16\\
28 & $(280,160,80,4,2,1)$ & 11 & 276 & $(272,96,2,1)$ & 16\\
29 & $(160,68,16,8,2,1)$ & 15 & 277 & $(296,96,2,1)$ & 15\\
30 & $(256,68,16,8,2,1)$ & 12 & 278 & $(304,96,2,1)$ & 16\\
31 & $(68,32,20,8,2,1)$ & 15 & 279 & $(68,16,10,1)$ & 17\\
32 & $(128,32,20,8,2,1)$ & 12 & 280 & $(96,16,10,1)$ & 16\\
33 & $(160,68,20,8,2,1)$ & 14 & 281 & $(256,16,10,1)$ & 16\\
34 & $(256,68,20,8,2,1)$ & 14 & 282 & $(258,16,10,1)$ & 17\\
35 & $(128,96,20,8,2,1)$ & 15 & 283 & $(32,20,10,1)$ & 16\\
36 & $(256,96,20,8,2,1)$ & 13 & 284 & $(68,20,10,1)$ & 17\\
37 & $(256,128,20,8,2,1)$ & 12 & 285 & $(96,20,10,1)$ & 17\\
38 & $(288,128,20,8,2,1)$ & 12 & 286 & $(128,20,10,1)$ & 16\\
39 & $(256,160,20,8,2,1)$ & 12 & 287 & $(160,20,10,1)$ & 17\\
40 & $(128,84,32,8,2,1)$ & 12 & 288 & $(256,20,10,1)$ & 16\\
41 & $(384,84,32,8,2,1)$ & 12 & 289 & $(258,20,10,1)$ & 17\\
42 & $(256,128,32,8,2,1)$ & 9 & 290 & $(68,32,10,1)$ & 17\\
43 & $(272,128,32,8,2,1)$ & 12 & 291 & $(84,32,10,1)$ & 16\\
44 & $(272,132,32,8,2,1)$ & 12 & 292 & $(128,32,10,1)$ & 16\\
45 & $(256,160,68,8,2,1)$ & 13 & 293 & $(132,32,10,1)$ & 16\\
46 & $(288,160,68,8,2,1)$ & 12 & 294 & $(258,32,10,1)$ & 16\\
47 & $(304,160,68,8,2,1)$ & 14 & 295 & $(272,32,10,1)$ & 17\\
48 & $(256,160,84,8,2,1)$ & 12 & 296 & $(320,32,10,1)$ & 16\\
49 & $(256,132,96,8,2,1)$ & 12 & 297 & $(336,32,10,1)$ & 16\\
50 & $(272,132,96,8,2,1)$ & 13 & 298 & $(384,32,10,1)$ & 16\\
51 & $(148,80,32,12,2,1)$ & 14 & 299 & $(386,32,10,1)$ & 16\\
52 & $(256,80,32,12,2,1)$ & 12 & 300 & $(160,68,10,1)$ & 17\\
53 & $(384,80,32,12,2,1)$ & 12 & 301 & $(176,68,10,1)$ & 17\\
54 & $(256,132,32,12,2,1)$ & 12 & 302 & $(304,68,10,1)$ & 17\\
55 & $(272,132,32,12,2,1)$ & 14 & 303 & $(160,84,10,1)$ & 16\\
56 & $(256,160,80,12,2,1)$ & 14 & 304 & $(256,84,10,1)$ & 16\\
57 & $(272,132,96,12,2,1)$ & 14 & 305 & $(258,84,10,1)$ & 17\\
58 & $(256,152,80,32,2,1)$ & 12 & 306 & $(132,96,10,1)$ & 17\\
59 & $(264,152,80,32,2,1)$ & 12 & 307 & $(272,96,10,1)$ & 17\\
60 & $(256,156,80,32,2,1)$ & 12 & 308 & $(288,96,10,1)$ & 16\\
61 & $(264,156,80,32,2,1)$ & 12 & 309 & $(290,96,10,1)$ & 17\\
62 & $(272,156,80,32,2,1)$ & 12 & 310 & $(304,96,10,1)$ & 16\\
63 & $(256,160,68,16,10,1)$ & 14 & 311 & $(384,96,10,1)$ & 16\\
64 & $(258,160,68,16,10,1)$ & 15 & 312 & $(386,96,10,1)$ & 17\\
65 & $(260,160,68,16,10,1)$ & 15 & 313 & $(256,160,10,1)$ & 16\\
66 & $(260,164,68,16,10,1)$ & 12 & 314 & $(258,160,10,1)$ & 17\\
67 & $(262,164,68,16,10,1)$ & 15 & 315 & $(260,160,10,1)$ & 17\\
68 & $(288,164,68,16,10,1)$ & 14 & 316 & $(262,160,10,1)$ & 17\\
69 & $(256,164,96,16,10,1)$ & 15 & 317 & $(288,160,10,1)$ & 15\\
70 & $(258,164,96,16,10,1)$ & 13 & 318 & $(290,160,10,1)$ & 17\\
71 & $(128,68,32,20,10,1)$ & 14 & 319 & $(292,160,10,1)$ & 16\\
72 & $(258,68,32,20,10,1)$ & 15 & 320 & $(304,160,10,1)$ & 16\\
73 & $(322,192,32,20,10,1)$ & 12 & 321 & $(306,160,10,1)$ & 17\\
74 & $(324,192,32,20,10,1)$ & 14 & 322 & $(308,160,10,1)$ & 17\\
75 & $(262,160,68,20,10,1)$ & 16 & 323 & $(192,36,16,1)$ & 15\\
76 & $(288,160,68,20,10,1)$ & 15 & 324 & $(196,36,16,1)$ & 16\\
77 & $(260,128,96,20,10,1)$ & 15 & 325 & $(448,36,16,1)$ & 15\\
78 & $(262,128,96,20,10,1)$ & 14 & 326 & $(450,36,16,1)$ & 15\\
79 & $(258,160,96,20,10,1)$ & 14 & 327 & $(132,96,16,1)$ & 15\\
80 & $(260,160,96,20,10,1)$ & 15 & 328 & $(140,96,16,1)$ & 16\\
81 & $(290,160,96,20,10,1)$ & 15 & 329 & $(164,96,16,1)$ & 16\\
82 & $(292,160,96,20,10,1)$ & 14 & 330 & $(258,96,16,1)$ & 16\\
83 & $(274,144,68,32,10,1)$ & 14 & 331 & $(290,96,16,1)$ & 15\\
84 & $(258,128,84,32,10,1)$ & 14 & 332 & $(296,96,16,1)$ & 16\\
85 & $(262,131,68,35,20,10)$ & 17 & 333 & $(298,96,16,1)$ & 16\\
86 & $(16,8,4,2,1)$ & 12 & 334 & $(140,98,16,1)$ & 16\\
87 & $(64,8,4,2,1)$ & 11 & 335 & $(164,98,16,1)$ & 16\\
88 & $(80,8,4,2,1)$ & 13 & 336 & $(256,98,16,1)$ & 16\\
89 & $(128,8,4,2,1)$ & 12 & 337 & $(258,98,16,1)$ & 16\\
90 & $(160,8,4,2,1)$ & 12 & 338 & $(288,98,16,1)$ & 16\\
91 & $(152,80,4,2,1)$ & 12 & 339 & $(296,98,16,1)$ & 16\\
92 & $(160,80,4,2,1)$ & 13 & 340 & $(258,100,16,1)$ & 16\\
93 & $(68,16,8,2,1)$ & 15 & 341 & $(290,100,16,1)$ & 16\\
94 & $(256,16,8,2,1)$ & 12 & 342 & $(296,100,16,1)$ & 16\\
95 & $(32,20,8,2,1)$ & 15 & 343 & $(298,100,16,1)$ & 16\\
96 & $(68,20,8,2,1)$ & 16 & 344 & $(386,100,16,1)$ & 16\\
97 & $(96,20,8,2,1)$ & 15 & 345 & $(416,100,16,1)$ & 15\\
98 & $(128,20,8,2,1)$ & 15 & 346 & $(424,100,16,1)$ & 16\\
99 & $(160,20,8,2,1)$ & 15 & 347 & $(258,228,16,1)$ & 16\\
100 & $(256,20,8,2,1)$ & 14 & 348 & $(262,228,16,1)$ & 16\\
101 & $(84,32,8,2,1)$ & 15 & 349 & $(292,228,16,1)$ & 15\\
102 & $(128,32,8,2,1)$ & 12 & 350 & $(294,228,16,1)$ & 16\\
103 & $(132,32,8,2,1)$ & 15 & 351 & $(302,228,16,1)$ & 16\\
104 & $(384,32,8,2,1)$ & 12 & 352 & $(136,38,20,1)$ & 16\\
105 & $(160,68,8,2,1)$ & 15 & 353 & $(196,38,20,1)$ & 17\\
106 & $(256,68,8,2,1)$ & 14 & 354 & $(136,96,20,1)$ & 17\\
107 & $(272,68,8,2,1)$ & 15 & 355 & $(164,96,20,1)$ & 17\\
108 & $(160,84,8,2,1)$ & 15 & 356 & $(290,96,20,1)$ & 17\\
109 & $(256,84,8,2,1)$ & 14 & 357 & $(298,96,20,1)$ & 17\\
110 & $(132,96,8,2,1)$ & 15 & 358 & $(386,96,20,1)$ & 17\\
111 & $(256,96,8,2,1)$ & 14 & 359 & $(424,96,20,1)$ & 17\\
112 & $(272,96,8,2,1)$ & 15 & 360 & $(160,100,20,1)$ & 16\\
113 & $(384,96,8,2,1)$ & 14 & 361 & $(298,100,20,1)$ & 17\\
114 & $(256,160,8,2,1)$ & 14 & 362 & $(296,102,20,1)$ & 17\\
115 & $(288,160,8,2,1)$ & 12 & 363 & $(262,160,20,1)$ & 17\\
116 & $(304,160,8,2,1)$ & 14 & 364 & $(264,160,20,1)$ & 16\\
117 & $(80,32,12,2,1)$ & 14 & 365 & $(266,160,20,1)$ & 17\\
118 & $(132,32,12,2,1)$ & 15 & 366 & $(268,160,20,1)$ & 17\\
119 & $(256,32,12,2,1)$ & 14 & 367 & $(270,160,20,1)$ & 17\\
120 & $(272,32,12,2,1)$ & 15 & 368 & $(322,160,20,1)$ & 17\\
121 & $(384,32,12,2,1)$ & 14 & 369 & $(326,160,20,1)$ & 17\\
122 & $(148,80,12,2,1)$ & 15 & 370 & $(334,160,20,1)$ & 16\\
123 & $(160,80,12,2,1)$ & 15 & 371 & $(290,224,20,1)$ & 17\\
124 & $(256,80,12,2,1)$ & 15 & 372 & $(298,224,20,1)$ & 17\\
125 & $(132,96,12,2,1)$ & 15 & 373 & $(300,224,20,1)$ & 17\\
126 & $(256,96,12,2,1)$ & 14 & 374 & $(282,160,84,1)$ & 17\\
127 & $(272,96,12,2,1)$ & 15 & 375 & $(298,160,84,1)$ & 16\\
128 & $(384,96,12,2,1)$ & 15 & 376 & $(160,68,19,10)$ & 17\\
129 & $(288,132,12,2,1)$ & 15 & 377 & $(161,68,19,10)$ & 17\\
130 & $(256,160,12,2,1)$ & 15 & 378 & $(162,68,19,10)$ & 17\\
131 & $(288,160,12,2,1)$ & 14 & 379 & $(164,68,19,10)$ & 17\\
132 & $(304,160,12,2,1)$ & 15 & 380 & $(258,68,19,10)$ & 17\\
133 & $(152,80,32,2,1)$ & 14 & 381 & $(261,68,19,10)$ & 16\\
134 & $(156,80,32,2,1)$ & 15 & 382 & $(262,68,19,10)$ & 17\\
135 & $(256,80,32,2,1)$ & 12 & 383 & $(289,68,19,10)$ & 17\\
136 & $(264,80,32,2,1)$ & 14 & 384 & $(68,35,20,10)$ & 17\\
137 & $(272,80,32,2,1)$ & 14 & 385 & $(160,68,20,10)$ & 17\\
138 & $(400,80,32,2,1)$ & 12 & 386 & $(161,68,20,10)$ & 17\\
139 & $(408,80,32,2,1)$ & 14 & 387 & $(162,68,20,10)$ & 17\\
140 & $(256,84,32,2,1)$ & 14 & 388 & $(164,68,20,10)$ & 17\\
141 & $(264,84,32,2,1)$ & 15 & 389 & $(166,68,20,10)$ & 17\\
142 & $(272,84,32,2,1)$ & 14 & 390 & $(167,68,20,10)$ & 17\\
143 & $(400,84,32,2,1)$ & 14 & 391 & $(258,68,20,10)$ & 16\\
144 & $(408,84,32,2,1)$ & 15 & 392 & $(259,68,20,10)$ & 17\\
145 & $(280,160,80,2,1)$ & 14 & 393 & $(262,68,20,10)$ & 17\\
146 & $(280,160,84,2,1)$ & 15 & 394 & $(289,68,20,10)$ & 17\\
147 & $(160,68,16,10,1)$ & 16 & 395 & $(131,96,20,10)$ & 17\\
148 & $(164,68,16,10,1)$ & 15 & 396 & $(133,96,20,10)$ & 17\\
149 & $(256,68,16,10,1)$ & 15 & 397 & $(135,96,20,10)$ & 17\\
150 & $(258,68,16,10,1)$ & 16 & 398 & $(161,96,20,10)$ & 17\\
151 & $(288,68,16,10,1)$ & 15 & 399 & $(164,96,20,10)$ & 17\\
152 & $(164,96,16,10,1)$ & 16 & 400 & $(259,96,20,10)$ & 17\\
153 & $(256,96,16,10,1)$ & 15 & 401 & $(261,96,20,10)$ & 17\\
154 & $(288,96,16,10,1)$ & 14 & 402 & $(391,96,20,10)$ & 17\\
155 & $(290,96,16,10,1)$ & 15 & 403 & $(304,161,68,10)$ & 17\\
156 & $(68,32,20,10,1)$ & 16 & 404 & $(309,162,68,10)$ & 17\\
157 & $(128,32,20,10,1)$ & 15 & 405 & $(279,178,68,10)$ & 16\\
158 & $(192,32,20,10,1)$ & 15 & 406 & $(276,179,68,10)$ & 17\\
159 & $(160,68,20,10,1)$ & 16 & 407 & $(257,163,84,10)$ & 17\\
160 & $(258,68,20,10,1)$ & 16 & 408 & $(262,163,84,10)$ & 17\\
161 & $(288,68,20,10,1)$ & 16 & 409 & $(276,165,96,10)$ & 17\\
162 & $(128,96,20,10,1)$ & 16 & 410 & $(4,2,1)$ & 15\\
163 & $(160,96,20,10,1)$ & 15 & 411 & $(8,2,1)$ & 16\\
164 & $(256,96,20,10,1)$ & 16 & 412 & $(12,2,1)$ & 17\\
165 & $(288,96,20,10,1)$ & 15 & 413 & $(32,2,1)$ & 16\\
166 & $(290,96,20,10,1)$ & 16 & 414 & $(80,2,1)$ & 17\\
167 & $(258,128,20,10,1)$ & 15 & 415 & $(84,2,1)$ & 17\\
168 & $(260,128,20,10,1)$ & 15 & 416 & $(96,2,1)$ & 17\\
169 & $(288,128,20,10,1)$ & 15 & 417 & $(16,10,1)$ & 17\\
170 & $(292,128,20,10,1)$ & 14 & 418 & $(20,10,1)$ & 17\\
171 & $(294,128,20,10,1)$ & 15 & 419 & $(32,10,1)$ & 17\\
172 & $(256,160,20,10,1)$ & 15 & 420 & $(68,10,1)$ & 18\\
173 & $(258,160,20,10,1)$ & 15 & 421 & $(84,10,1)$ & 17\\
174 & $(260,160,20,10,1)$ & 15 & 422 & $(96,10,1)$ & 17\\
175 & $(262,160,20,10,1)$ & 16 & 423 & $(160,10,1)$ & 17\\
176 & $(128,68,32,10,1)$ & 15 & 424 & $(256,10,1)$ & 17\\
177 & $(144,68,32,10,1)$ & 15 & 425 & $(258,10,1)$ & 18\\
178 & $(272,68,32,10,1)$ & 16 & 426 & $(272,10,1)$ & 18\\
179 & $(128,84,32,10,1)$ & 15 & 427 & $(36,16,1)$ & 17\\
180 & $(144,84,32,10,1)$ & 15 & 428 & $(96,16,1)$ & 17\\
181 & $(258,84,32,10,1)$ & 15 & 429 & $(98,16,1)$ & 17\\
182 & $(272,84,32,10,1)$ & 15 & 430 & $(100,16,1)$ & 17\\
183 & $(384,84,32,10,1)$ & 15 & 431 & $(228,16,1)$ & 17\\
184 & $(386,84,32,10,1)$ & 15 & 432 & $(256,16,1)$ & 16\\
185 & $(400,84,32,10,1)$ & 15 & 433 & $(258,16,1)$ & 17\\
186 & $(258,128,32,10,1)$ & 15 & 434 & $(266,16,1)$ & 17\\
187 & $(262,128,32,10,1)$ & 15 & 435 & $(38,20,1)$ & 17\\
188 & $(274,128,32,10,1)$ & 15 & 436 & $(96,20,1)$ & 17\\
189 & $(278,128,32,10,1)$ & 15 & 437 & $(100,20,1)$ & 17\\
190 & $(324,128,32,10,1)$ & 15 & 438 & $(102,20,1)$ & 17\\
191 & $(326,128,32,10,1)$ & 14 & 439 & $(160,20,1)$ & 17\\
192 & $(340,128,32,10,1)$ & 15 & 440 & $(224,20,1)$ & 17\\
193 & $(342,128,32,10,1)$ & 15 & 441 & $(258,20,1)$ & 17\\
194 & $(262,132,32,10,1)$ & 16 & 442 & $(264,20,1)$ & 17\\
195 & $(274,132,32,10,1)$ & 15 & 443 & $(266,20,1)$ & 17\\
196 & $(278,132,32,10,1)$ & 16 & 444 & $(160,84,1)$ & 17\\
197 & $(324,132,32,10,1)$ & 15 & 445 & $(266,84,1)$ & 17\\
198 & $(326,132,32,10,1)$ & 15 & 446 & $(304,160,1)$ & 17\\
199 & $(336,132,32,10,1)$ & 15 & 447 & $(306,160,1)$ & 17\\
200 & $(338,132,32,10,1)$ & 15 & 448 & $(314,160,1)$ & 17\\
201 & $(340,132,32,10,1)$ & 16 & 449 & $(316,160,1)$ & 17\\
202 & $(304,160,68,10,1)$ & 16 & 450 & $(68,19,10)$ & 18\\
203 & $(272,176,68,10,1)$ & 16 & 451 & $(257,19,10)$ & 18\\
204 & $(274,176,68,10,1)$ & 16 & 452 & $(35,20,10)$ & 18\\
205 & $(276,176,68,10,1)$ & 16 & 453 & $(68,20,10)$ & 18\\
206 & $(256,160,84,10,1)$ & 14 & 454 & $(96,20,10)$ & 18\\
207 & $(258,160,84,10,1)$ & 15 & 455 & $(129,20,10)$ & 18\\
208 & $(260,160,84,10,1)$ & 15 & 456 & $(161,20,10)$ & 18\\
209 & $(288,160,84,10,1)$ & 15 & 457 & $(449,20,10)$ & 18\\
210 & $(290,160,84,10,1)$ & 15 & 458 & $(450,20,10)$ & 18\\
211 & $(292,160,84,10,1)$ & 15 & 459 & $(160,68,10)$ & 18\\
212 & $(274,132,96,10,1)$ & 16 & 460 & $(161,68,10)$ & 17\\
213 & $(276,132,96,10,1)$ & 16 & 461 & $(162,68,10)$ & 18\\
214 & $(278,132,96,10,1)$ & 16 & 462 & $(178,68,10)$ & 18\\
215 & $(292,132,96,10,1)$ & 14 & 463 & $(179,68,10)$ & 18\\
216 & $(294,132,96,10,1)$ & 16 & 464 & $(180,68,10)$ & 18\\
217 & $(308,132,96,10,1)$ & 16 & 465 & $(304,68,10)$ & 18\\
218 & $(310,132,96,10,1)$ & 16 & 466 & $(305,68,10)$ & 18\\
219 & $(324,192,36,16,1)$ & 12 & 467 & $(163,84,10)$ & 18\\
220 & $(334,192,36,16,1)$ & 15 & 468 & $(164,84,10)$ & 18\\
221 & $(330,196,36,16,1)$ & 15 & 469 & $(258,84,10)$ & 18\\
222 & $(262,132,96,16,1)$ & 15 & 470 & $(259,84,10)$ & 18\\
223 & $(294,132,96,16,1)$ & 15 & 471 & $(289,84,10)$ & 18\\
224 & $(302,132,96,16,1)$ & 15 & 472 & $(290,84,10)$ & 18\\
225 & $(262,140,96,16,1)$ & 15 & 473 & $(293,84,10)$ & 18\\
226 & $(290,140,96,16,1)$ & 15 & 474 & $(294,84,10)$ & 18\\
227 & $(294,140,96,16,1)$ & 15 & 475 & $(165,96,10)$ & 18\\
228 & $(262,164,96,16,1)$ & 15 & 476 & $(276,96,10)$ & 18\\
229 & $(266,164,96,16,1)$ & 15 & 477 & $(286,163,84)$ & 18\\
230 & $(270,164,96,16,1)$ & 15 & 478 & $(2,1)$ & 17\\
231 & $(256,140,98,16,1)$ & 15 & 479 & $(10,1)$ & 18\\
232 & $(258,140,98,16,1)$ & 15 & 480 & $(16,1)$ & 18\\
233 & $(260,140,98,16,1)$ & 15 & 481 & $(20,1)$ & 18\\
234 & $(264,164,98,16,1)$ & 15 & 482 & $(84,1)$ & 18\\
235 & $(326,136,38,20,1)$ & 16 & 483 & $(160,1)$ & 18\\
236 & $(262,136,96,20,1)$ & 16 & 484 & $(19,10)$ & 18\\
237 & $(266,136,96,20,1)$ & 16 & 485 & $(20,10)$ & 18\\
238 & $(294,136,96,20,1)$ & 16 & 486 & $(68,10)$ & 18\\
239 & $(262,160,68,19,10)$ & 16 & 487 & $(84,10)$ & 18\\
240 & $(261,161,68,19,10)$ & 16 & 488 & $(96,10)$ & 18\\
241 & $(258,162,68,19,10)$ & 16 & 489 & $(258,10)$ & 18\\
242 & $(258,164,68,19,10)$ & 16 & 490 & $(275,10)$ & 18\\
243 & $(131,68,35,20,10)$ & 17 & 491 & $(163,84)$ & 18\\
244 & $(133,68,35,20,10)$ & 16 & 492 & $(1)$ & 19\\
245 & $(262,161,68,20,10)$ & 16 & 493 & $(10)$ & 19\\
246 & $(259,162,68,20,10)$ & 16 & 494 & $(84)$ & 19\\
247 & $(259,135,96,20,10)$ & 17 & 495 & $()$ & 20\\
\end{longtable}
\endgroup

\clearpage
\section{An explicit 23-term witness}\label{app:upper}

For completeness, the following triples of row-major codes give a
length-$23$ decomposition in convention~\eqref{eq:tensor}. They are
the reduction modulo two of the integer scheme
\nolinkurl{gg-333-rank23-rec-0-0-0-z.txt} from the flip-cpd
collection of Khoruzhii, Gel{\ss}, and Pokutta~\cite{flip-cpd}.
Product order and row-major output coordinates are preserved.
Appendix~\ref{app:trajectory-sources} fixes the source version.
This is the witness used by the formal declaration
\code{rank23\_entry\_identity}.

\begin{center}
\begin{tabular}{rrrr@{\qquad}rrrr}
\toprule
$t$ & $A_t$ & $B_t$ & $C_t$ & $t$ & $A_t$ & $B_t$ & $C_t$\\
\midrule
1 & 400 & 304 & 418 & 13 & 284 & 264 & 12\\
2 & 260 & 193 & 10 & 14 & 9 & 67 & 66\\
3 & 73 & 7 & 64 & 15 & 256 & 432 & 130\\
4 & 393 & 12 & 68 & 16 & 146 & 40 & 4\\
5 & 292 & 448 & 8 & 17 & 448 & 4 & 320\\
6 & 8 & 2 & 210 & 18 & 7 & 8 & 5\\
7 & 265 & 72 & 65 & 19 & 32 & 128 & 24\\
8 & 146 & 16 & 2 & 20 & 384 & 52 & 288\\
9 & 269 & 65 & 3 & 21 & 268 & 9 & 9\\
10 & 72 & 2 & 192 & 22 & 56 & 256 & 40\\
11 & 392 & 260 & 36 & 23 & 16 & 16 & 432\\
12 & 144 & 288 & 390 & & & &\\
\bottomrule
\end{tabular}
\end{center}

To check the witness, decode each integer into its nine bits and compute
\[
 \sum_{t=1}^{23} A_t[i,j]B_t[u,v]C_t[w,z]\pmod2.
\]
For every choice of the six indices in $\{0,1,2\}$ this equals
$1$ precisely when $u=j$, $w=i$, and $z=v$, and is zero otherwise.
These are all $729$ tensor coordinates. This finite identity is checked
by reduction in Lean, and it can also be reproduced directly from the
table. Bilinearity then gives~\eqref{eq:algorithm} for all input pairs.
Together with Theorem~\ref{thm:main}, it proves the interval $[21,23]$.

\section{Formal interfaces and verification record}\label{app:formalization}

\subsection{Statements and dependencies}
The formal project uses Lean $4.33.1$ and Mathlib, with dependency
versions fixed in the accompanying source package.
Definitions use the tensor convention~\eqref{eq:tensor}.
The declaration
\code{tensorEntry\_identity\_iff\_bilinearAlgorithm}
proves its equivalence with~\eqref{eq:algorithm} on all inputs.
In the namespace \code{QiushiMatmul}, the two principal statements are
\[
 \begin{aligned}
 \code{rank\_ge\_21}&:\ \code{RankAtLeast 21},\\
 \code{rank\_between\_21\_and\_23}
   &:\ \code{RankAtLeast 21}\ \wedge\
       \code{TensorEntryRankAtMost 23}.
 \end{aligned}
\]
Here \code{RankAtLeast n} quantifies over every length $r$ and
every entrywise decomposition of that length, allowing zero summands,
and asserts $n\leq r$.

The closed definition \code{provedFinitePremises} supplies the line,
all-high plane, and normalized affine-plane bounds to
{\small\nolinkurl{rank_ge_21_of_premises}}.
The plane bounds are {\small\nolinkurl{plane484_rank_ge_19}} and
\code{plane485Gen\_lb19} through \code{plane491Gen\_lb19}.
The first is the named-space version of \code{plane484Gen\_lb19}
at \code{spanCodes [19, 10]}. Their proofs use separately established
source bounds and checked tensor transports.

\subsection{Locating the source bounds}\label{app:source-interfaces}
All paths in this subsection are relative to the accompanying
\nolinkurl{formalization/} directory~\cite{qiushi-lean}.

For the source-bound rules in Section~\ref{sec:source-bounds},
\nolinkurl{contraction_seed_invertible_minor} proves the contraction
bound from an explicit inverse minor; its multi-contraction versions
handle stacked flattenings.
\nolinkurl{weighted_cover_quotient_decomp_false} proves the cover
inequality, and \code{step81o8s2Span\_lb} is the source bound $9$
in~\eqref{eq:three-hyperplane-bound}.
\code{no\_short\_of\_deletions} is
Lemma~\ref{lem:source-substitution}; \code{orbit11\_lb12}
is its worked application to $W_{11}$.
The concrete weighted bound for $W_{450}$ is
\nolinkurl{step108_orbit450_lb18_unconditional}.
The interface \code{generic\_plane\_qra} combines source bounds,
quotient membership, zero-forcing coverage, and an integer-system
exclusion to obtain each plane raise.

The line-cover proof~\eqref{eq:line-cover} is in
\nolinkurl{QiushiLineCoverageFinal.lean}, with its $255$ planes,
labels, and membership tests in \nolinkurl{QiushiLineCoverageDefs.lean}
and the eight coverage chunks it imports. The bound-$17$ inputs are
\code{orbit478\_lb17\_mono} and \code{orbit479\_lb17\_graded},
transported by containment from the bounds for $W_{412}$ and $W_{417}$.
The four bound-$18$ inputs are \code{plane480Gen\_lb18} through
\code{plane483Gen\_lb18}. Combining them gives
\code{line\_rank1\_ge19} in \nolinkurl{QiushiFinalTheorem.lean}.
The rank-two and rank-three line transfers are in
\nolinkurl{QiushiLineFromPlane.lean}: the matrices of codes $25$ in
$W_{484}$ and $106$ in $W_{488}$ are sent to codes $17$ and $273$,
respectively.

The three affine-span proofs are collected in
\nolinkurl{QiushiRemainingPremises.lean} as
\code{affineHyperplane0\_rank\_ge\_17} through
\code{affineHyperplane2\_rank\_ge\_17}.
For $(272,4,2)$, \nolinkurl{QiushiMonoOrbit414From262.lean}
uses the bound in \nolinkurl{QiushiBranch262Extraction.lean}
and the transformation $X\mapsto P^TXQ^{-T}$ with matrix codes
$P=161$, $Q=84$; the preimages of the three generators have codes
$136,1,2$ in $W_{262}$.
For $(273,4,2)$ and $(272,4,1)$, the files
\nolinkurl{QiushiBranch415Extraction.lean} and
\nolinkurl{QiushiBranch416Extraction.lean} connect the quotient
decompositions to the $101$-row and $373$-row integer systems.
Their respective \nolinkurl{Dispatch} modules provide the source
bounds and membership proofs, and the \nolinkurl{NoModel} modules
assemble the branch refutations. The eleven further affine spans are
transported in \nolinkurl{QiushiFinitePremisesReducer.lean}.

\subsection{Catalogue and occupation interfaces}
In the namespace \code{FrozenRegistry}, the theorem
\code{all\_representatives} proves the listed bounds.
The declarations \code{L0\_rank\_sound} and
\code{coverage\_with\_rank\_bound} combine them with orbit coverage
and consistent labels, as in Proposition~\ref{prop:global}.

Supplementary declarations prove the fourteen two-plane orbits,
subspace counts, quotient/restriction equivalence, all
$233{,}680$ row bindings in the eight full occupation systems,
target-$19$ occupation controls, projections of $23$-term
decompositions, and exact singleton capacities.
In the namespace \code{FrozenOccupation},
\code{plane484\_no\_model} through \code{plane491\_no\_model}
exclude arbitrary integer models of~\eqref{eq:integer-system}.
These are separate from the tensor-rank statements and do not identify
occupation feasibility with decomposition existence.

\subsection{Recorded kernel checks}
All $13{,}438$ registered modules passed fresh-source compilation.
The $89$ terminal declarations and their dependencies were then replayed
in the following six groups:
\begin{center}
\begin{tabular}{lrr}
\toprule
Group & Roots & Replayed declarations\\
\midrule
Main theorem and global finite bounds & 6 & $143{,}943$\\
Complete calibration & 7 & $37{,}629$\\
Orbit classification and exact labels & 12 & $33{,}687$\\
Full occupation systems & 8 & $47{,}338$\\
Structural results and consequences & 51 & $36{,}421$\\
Boolean encodings & 5 & $6{,}713$\\
\bottomrule
\end{tabular}
\end{center}
All six groups passed. Their dependency closures overlap, so the
right-hand column is not additive.
Each run used the official
\code{Lean.Environment.replay} from an empty kernel environment,
with constructor and recursor checks~\cite{lean-replay}.
Root types and transitive axioms were also audited.
These executions recheck the proof using Lean's kernel; they do not
constitute a second implementation of the kernel.

\subsection{Certificate realizations}\label{app:certificate-realizations}
The main formal theorem uses integer branch certificates with
proved quotient semantics. In the Boolean realization, a singleton
cap $c_q$ is encoded by $x_q=\sum_{j=1}^{c_q}y_{qj}$, where the
$y_{qj}$ are Boolean copies. Every legal integer count is represented
by choosing exactly $x_q$ true copies; conversely, summing the copies
recovers the count. Each occupation row becomes a cardinality
constraint. Sequential counters~\cite{sinz} with disjoint
auxiliary-variable sets encode all rows simultaneously, including
the two inequalities imposing the total. Thus the integer system
and its Boolean encoding are equisatisfiable.

Seven quotients have only singleton caps $0$ and $1$; orbit $489$
has one direction with cap $2$, represented by two copies. The Lean
development proves the generic Boolean-copy and disjoint-counter
semantics as well as the concrete integer exclusions. The historical
CNF/DRAT certificates give a separate computational realization of
the eight exclusions: their exact DIMACS and DRAT files are checked
by the original proof package, whereas Lean checks the integer
branch certificates. Feasible target-$19$ occupation vectors and
projected $23$-term controls check consistency of the constraints;
they do not construct rank-$19$ quotient decompositions.

The
\href{https://github.com/Oxelra-AI/Qiushi-Engine-Matmul-Research/blob/lean-formalization/formalization/README.md}{build instructions}
and
\href{https://github.com/Oxelra-AI/Qiushi-Engine-Matmul-Research/blob/lean-formalization/formalization/verification-results.json}{verification record}
for the proof library~\cite{qiushi-lean} identify the same sources,
certificate data, and dependencies.
The manuscript's LaTeX package includes the finite catalogue and
verification summary, and can be built independently of the
formalization.

\subsection{Sources for the research trajectory}\label{app:trajectory-sources}
The following links fix the research materials underlying
Section~\ref{sec:development} to a single repository revision.
The manuscript data file \nolinkurl{data/artifact-sources.json}
also lists the upstream schemes and their coordinate conversions.

The numerical example uses the
\href{https://github.com/Oxelra-AI/Qiushi-Engine-Matmul-Research/blob/01b4a6baf3788bf3f14af98a653ea22dafa0691c/research/materials/exact_baselines/results/fmm_r23_schemes/puiseux_analysis.json}{pairing singular values}
and the
\href{https://github.com/Oxelra-AI/Qiushi-Engine-Matmul-Research/blob/01b4a6baf3788bf3f14af98a653ea22dafa0691c/research/materials/exact_baselines/results/fmm_r23_schemes/cancellation_analysis.json}{absorbed-factor data}.
The latter normalize a pairing column by a nonzero coordinate;
Table~\ref{tab:continuation} converts to unit Euclidean norm before
applying~\eqref{eq:deletion-diagnostic}. The target residuals are
the recorded \nolinkurl{solve.brent_residual_norm} values at
$t=10,20,60$ in the
\href{https://github.com/Oxelra-AI/Qiushi-Engine-Matmul-Research/blob/01b4a6baf3788bf3f14af98a653ea22dafa0691c/research/materials/exact_baselines/results/fmm_r23_schemes/serendipitous_8d34_continuation_ranked0_further.json}{saved continuation states},
which also contain all $621$ coefficients at each point.
The
\href{https://github.com/Oxelra-AI/Qiushi-Engine-Matmul-Research/blob/01b4a6baf3788bf3f14af98a653ea22dafa0691c/research/materials/deformation_and_incidence/notes/corrected_cancellation.md}{corrected cancellation analysis}
records the change in interpretation. The manuscript source includes
the selected values and the arithmetic used for the displayed table.

The continuation starts from the
\href{https://github.com/dronperminov/FastMatrixMultiplication/blob/e703f18cbc157cec8ef68cc60bc9de0a149661ae/schemes/results/serendipitous_base/3x3x3_m23_8d34d377660f8f8d8b32cd4b6a1e1c40a5093dd0_ZT.json}{integer scheme in Perminov's collection}~\cite{perminov-code}.
The first two coefficient vectors are retained; the third is
reindexed by $w^{\rm new}_{3i+k}=w^{\rm source}_{3k+i}$.
The
\href{https://github.com/Oxelra-AI/Qiushi-Engine-Matmul-Research/blob/01b4a6baf3788bf3f14af98a653ea22dafa0691c/research/materials/exact_baselines/code/fmm_to_qmm.py}{conversion program}
and the
\href{https://github.com/Oxelra-AI/Qiushi-Engine-Matmul-Research/blob/01b4a6baf3788bf3f14af98a653ea22dafa0691c/research/materials/exact_baselines/results/fmm_r23_schemes/serendipitous_8d34.qmm}{converted scheme}
specify the exact starting point. Appendix~\ref{app:upper} uses a
different construction, the
\href{https://github.com/khoruzhii/flip-cpd/blob/9eeb17f487f14103d7128cf290ae61535d8983b1/data/schemes_paper/gg-333-rank23-rec-0-0-0-z.txt}{general $3\times3$ integer scheme in flip-cpd}~\cite{flip-cpd}.
Reducing its coefficients modulo two and reading
$a_1,\ldots,a_9$, $b_1,\ldots,b_9$, and $c_1,\ldots,c_9$
in row-major order gives exactly the printed triples.

The
\href{https://github.com/Oxelra-AI/Qiushi-Engine-Matmul-Research/blob/01b4a6baf3788bf3f14af98a653ea22dafa0691c/research/materials/quotient_cores/notes/core_bridge_review.md}{core-bridge analysis}
explains the relation between quotient construction and full-tensor
bounds. The
\href{https://github.com/Oxelra-AI/Qiushi-Engine-Matmul-Research/blob/01b4a6baf3788bf3f14af98a653ea22dafa0691c/research/materials/finite_certification/notes/seqcounter_error.md}{counter correction}
isolates the incompatible auxiliary-variable allocation and checks
the corrected cardinality constraint. The
\href{https://github.com/Oxelra-AI/Qiushi-Engine-Matmul-Research/blob/01b4a6baf3788bf3f14af98a653ea22dafa0691c/research/materials/structural_obstruction/notes/short_saturated_lower_bound_route.md}{shortened saturation proof}
records the replacement of the terminal support exclusions by saturation
and the removal of the proposed orbit-$479$ strengthening from the
coset reduction. These notes describe research states at
the time they were written; the theorem and its final dependencies are
given in the present paper and formalization.

\bibliographystyle{unsrturl}
\bibliography{references}
\end{document}